\documentclass{article}
\usepackage{geometry}                % See geometry.pdf to learn the layout options. There are lots.
\usepackage{graphicx}
\usepackage{stmaryrd}
\usepackage{amssymb}
\usepackage{amsthm, bm}
\usepackage{amsmath, cancel, centernot }
\usepackage[mathscr]{eucal}
\usepackage{mathtools}
\usepackage{epstopdf}
\title{Discretisation of the Bloch Sphere, Fractal Invariant Sets and Bell's Theorem}
\author{T.N.Palmer\\ Department of Physics, University of Oxford, UK\\
tim.palmer@physics.ox.ac.uk}
\date{\today}                                          % Activate to display a given date or no date
\makeatletter
\newcommand\be{\@ifstar{\[}{\begin{equation}}}
\newcommand\ee{\@ifstar{\]}{\end{equation}}}
\newcommand\bp{\begin{pmatrix}}
\newcommand\ep{\end{pmatrix}}

\newtheorem{theorem}{Theorem}

\newtheorem{lemma}{Lemma}

\makeatother
\begin{document}
\bibliographystyle{plain}
\maketitle

\begin{abstract}
An arbitrarily dense discretisation of the Bloch sphere of complex Hilbert states is constructed, where points correspond to bit strings of fixed finite length. Number-theoretic properties of trigonometric functions (not part of the quantum-theoretic canon) are used to show that this constructive discretised representation incorporates many of the defining characteristics of quantum systems: completementarity, uncertainty relationships and (with a simple Cartesian product of discretised spheres) entanglement. Unlike Meyer's earlier discretisation of the Bloch Sphere, there are no orthonormal triples, hence the Kocken-Specker theorem is not nullified. A physical interpretation of points on the discretised Bloch sphere is given in terms of ensembles of trajectories on a dynamically invariant fractal set in state space, where states of physical reality correspond to points on the invariant set. This deterministic construction provides a new way to understand the violation of the Bell inequality without violating statistical independence or factorisation, where these conditions are defined solely from states on the invariant set. In this finite representation there is an upper limit to the number of qubits that can be entangled, a property with potential experimental consequences.  
\end{abstract}

\section{Introduction}

The fields $\mathbb R$ and $\mathbb C$ are deeply embedded in the formalism of both classical and quantum theories of physics. However, the status of these continuum fields is fundamentally different in the two classes of theory. 

Consider, for example, a typical finite dimensional classical system with well-posed initial value-problem, i.e. where the state of the system at time $t>t_0$ depends continuously on the initial conditions at $t_0$ (e.g. the chaotic Lorenz equations \cite{Lorenz:1963} \cite{Palmer:2014b}). This property of continuity ensures that the initial-value problem can be solved to arbitrary accuracy by algorithm. In this sense the strict continuum (with all its inherent paradoxes \cite{Wagon}) plays no essential role in classical theory, and the real line can, if desired, be considered as approximating some underpinning granular structure in the smooth limit as such granularity tends to zero. 

By contrast, the continuum plays a vital role in quantum theory (even the quantum theory of finite-dimensional systems). This point was made explicitly in Hardy's axiomatic approach to quantum theory \cite{Hardy:2004} (see also \cite{MasanesMuller} \cite{Chiribella}). In particular, Hardy's `Continuity Axiom' states: `There exists a continuous reversible transformation (one which can be made up from $\ldots$ transformations only infinitesimally different from the identity) of a system between any two pure states of that system.' As Hardy notes, the Continuity Axiom provides the key difference between classical and quantum theory. As such, the continuum complex Hilbert space of quantum theory (and hence quantum theory itself) can only be a singular \cite{Berry} and not a smooth limit of some more granular representation of quantum physics, should such a description exist. 

As an example which makes the role of  Hardy's axiom explicit - one that will be the focus of this paper - consider a discretisation of the set of complex Hilbert states
\be
\label{psi}
|\psi\rangle = \cos \frac{\theta}{2} |0\rangle + e^{i\phi} \sin \frac{\theta}{2}  |1\rangle
\ee
where $\cos^2 \theta/2$ and $\phi /2\pi$ are describable by a finite number $N$ of binary digits, and $N$ is some arbitrarily large but finite positive integer. Now consider a (Hadamard-like) unitary transform which maps
\be
\label{psi1}
 \cos \frac{\phi}{2} |0\rangle + \sin \frac{\phi}{2} |1\rangle \mapsto \frac{1}{\sqrt 2} ( |0\rangle + e^{i \phi} |1\rangle)
\ee
Suppose $0 < \phi <\pi/2$. In quantum theory, such a transform, which maps a point on the Bloch sphere to a second, is well defined. However, restricted to the discretisation defined above, the transform is undefined, no matter how big is $N$, for the following reason: by Niven's Theorem (see Section \ref{finite}), if $\cos^2 \phi/2$ can be described by a finite number of bits, then $\phi/2\pi$ cannot. Hence if the initial state of the transform belongs to the discretised set, then the final state does not, and \emph{vice versa}. 

Hardy's Continuity Axiom raises a profound question: should one simply acknowledge quantum theory with its continuum Hilbert state space as fundamentally correct and abandon any hope of describing physics finitely \cite{Ellis}, or should one seek an alternative finite-$N$ theory of quantum physics, necessarily different in theoretical structure to quantum theory, from which quantum theory is emergent only as a singular limit at $N=\infty$? In this paper we pursue the latter possibility. 

In Section \ref{finite} a constructive finite representation of complex Hilbert vectors is presented, where complex numbers and quaternions are linked to bit-string permutation/negation operators. It is shown how two of the most important defining properties of quantum theory: complementarity and the uncertainty relationships, derive from geometric properties of spherical triangles and number-theoretic properties of trigonometric functions. A physical interpretation of such finite Hilbert vectors is presented in Section \ref{physics} in terms of a symbolic representation of a deterministic fractal geometry $I_U$ (the invariant set) in cosmological state space, on which states of physical reality exist and evolve. A specific pedagogical example, the sequential Stern-Gerlach device, is discussed within the framework of this invariant set model to illustrate the non-commutativity of spin operators from number-theoretic properties of spherical triangles. In Section \ref{entangle}, the construction is generalised to include multiple entangled qubits. Unlike in quantum theory, the state space of $J$ qubits in this finite framework is simply the $J$-times Cartesian product of the single finite Bloch sphere. The Bell Theorem is analysed within this framework. In invariant set theory, the Bell inequality is violated exactly as quantum theory. Consistent with this, the invariant set model violates both the Statistical Independence and Factorisation postulates as they are usually defined in the Bell Theorem. However, by restricting to states which lie on $I_U$, it is argued that the Bell Inequality can be violated as in quantum theory, without violating `Statistical Independence on $I_U$' and 'Factorisation on $I_U$'. It then becomes a matter of definition whether invariant set theory violates free choice and local causality.  We argue that if the latter are based on processes occurring in space-time, the latter are plausible definitions of free choice and local causality, when the invariant set model is locally causal. 

This work can be considered a continuation of a programme started by Meyer \cite{Meyer:1999} who showed it was possible to nullify the Kochen-Specker theorem by considering a rational set of orthonormal triples of 'colourable' points which are dense on $\mathbb S^2$, each triple corresponding to an orthogonal vector triad in $\mathbb R^3$. The construction here shares some number-theoretic similarities with that of Meyer. However, in the construction here, as discussed in Section \ref{finite}, \emph{no} such triples of points exist. In this sense, and in contrast with that of Meyer, the present construction does not nullify the Kochen-Specker theorem (indeed, as discussed, the construction is explicitly contextual in character). 

In Section \ref{discussion} is discussed a possible property of this ansatz (which sets it apart from quantum theory) that could be tested experimentally: that only a finite number $\log_2 N$ of qubits can be mutually entangled.  Such a limit may conceivably be a manifestation in the invariant set model of the inherently decoherent nature of gravity. 

\section{From Finite Bit Strings to Hilbert Vectors}
\label{finite}
In this Section, a discretisation of the Bloch Sphere is constructed. In this representation, points on the Bloch Sphere represent finite bit strings.  In Section \ref{bit} it is shown how complex number and quaternionic structure emerge naturally from negation/permutation operators acting on such bit strings. Quantum-like properties of these bit strings are a consequence of the geometry of spherical triangles and number-theoretic properties of trigonometric functions, reviewed in Section \ref{number}, which are not part of the quantum theoretic canon. The mapping of bit strings onto the sphere is described in Section \ref{mapping}. The relationship between bit strings and discrete complex Hilbert vectors is described in Section \ref{hilbert}. The relation to Meyer's \cite{Meyer:1999} construction is outlined in Section \ref{meyer}.  

\subsection{Quaternions from Bit-String Permutations}
\label{bit}
Consider the bit string $\{a_1 \; a_2 \; a_3  \ldots a_{N}\}$ where $a_i \in \{a, \cancel a\}$ denote symbolic labels (to be defined), $N=2^M$ and $M \ge 2$ is an integer.  An order-$N$ cyclic permutation operator is defined as 
\be
\label{zeta}
\zeta \{a_1 \; a_2 \; a_3  \ldots a_{N}\}  = \{a_2\; a_3 \ldots a_N\; a_1\}
\ee
With this define
\begin{align}
T_a(n_2,n_1)&=\zeta^{n_1} \{\underbrace{aa \ldots a}_{N-n_2} \ \  \underbrace{\cancel a \cancel a \ldots \cancel a}_{n_2} \}
\end{align}
Then,
\begin{align}
\label{sprops}
T_a(\frac{N}{2}, \frac{N}{4})
&= \zeta^{\frac{N}{4}}\{\underbrace{a a \ldots a}_{\frac{N}{2}}\ \ \ \underbrace{\cancel a \cancel a \ldots \cancel a}_{\frac{N}{2}}\}
\nonumber \\
&= \{\underbrace{a  a \ldots  a}_{\frac{N}{4}}\ \ \ \underbrace{\cancel a \cancel a \ldots \cancel a}_{\frac{N}{4}}\} \ || \ \{\underbrace{\cancel a \cancel a \ldots \cancel a}_{\frac{N}{4}}\ \ \underbrace{a a \ldots a}_{\frac{N}{4}} \}
\end{align}
where $||$ denotes concatenation. 

We now introduce the `square-root-of-minus-one' operator $i$ defined by
\begin{align}
\label{i}
i \{a_1 a_2 \ldots a_{\frac{N}{2}}\}& =  \{a_{\frac{N}{4}+1}  a_{\frac{N}{4}+2} \ldots a_{\frac{N}{2}} \ \ \cancel a_1 \cancel a_2 \ldots \cancel a_{\frac{N}{4}} \} \nonumber \\
\implies i^2 \{a_1 a_2 \ldots a_{\frac{N}{2}}\}&=\ \{\cancel a_{1} \cancel a_{2} \ldots \cancel a_{\frac{N}{2}}\} \equiv -\{a_1 a_2 \ldots a_{\frac{N}{2}}\} 
\end{align}
where $\cancel a_i=\cancel a$ if $a_i=a$ and \emph{vice versa}. We also introduce the notation 
\be
\label{2comp}
T_1 || T_2 \equiv
\begin{pmatrix}
T_1\\
T_2
\end{pmatrix}
\ee
where $T_1$ and $T_2$ are each $N/2$-element bit strings. Using this, it is straightforward to show that
\begin{align}
\label{q}
T_a(\frac{N}{2}, 0)
= & 
\begin{pmatrix}
0 & 1 \\
-1 & 0
\end{pmatrix}
T_a(0,0) \nonumber \\
T_a(\frac{N}{2}, \frac{N}{4})=&
 \begin{pmatrix}
i & 0 \\
0 & -i
\end{pmatrix}
T_a(0,0)
= 
\begin{pmatrix}
0 &- i \\
-i & 0
\end{pmatrix}
T_a(\frac{N}{2}, 0).
\end{align}
In (\ref{q}), the three $ 2 \times 2$ matrices
\be 
\label{quat}
\mathbf{i}_1=
\begin{pmatrix}
0 & 1 \\
-1 & 0
\end{pmatrix};
\ \ \ 
\mathbf{i}_2=
 \begin{pmatrix}
i & 0 \\
0 & -i
\end{pmatrix}
\ \ \ \ 
\mathbf{i}_3=
\begin{pmatrix}
0 &- i \\
-i & 0
\end{pmatrix};
\ee
are pure unit quaternions (hence $\mathbf{i}^2_1=\mathbf{i}^2_2=\mathbf{i}^2_3=\mathbf{i}_1 \circ \mathbf{i}_2 \circ \mathbf{i}_3=-\mathbf{1})$, where $\mathbf{1}$ is the identity matrix.
Each of these quaternions is itself a pure unit complex number and therefore can be thought of as representing a rotation by $\pi/2$ radians. Since the quaternions represent rotations about orthogonal directions in 3-space, we can represent the relations (\ref{q}) graphically, as shown in Fig \ref{quatern}. The direction of the axis labelled $\hat a$ is arbitrary. Before continuing with mapping $T_a(m, n)$ onto the sphere, we need a key number theoretic result.

%..........................................................................................................................................
\begin{figure}
\centering
\includegraphics[scale=0.3]{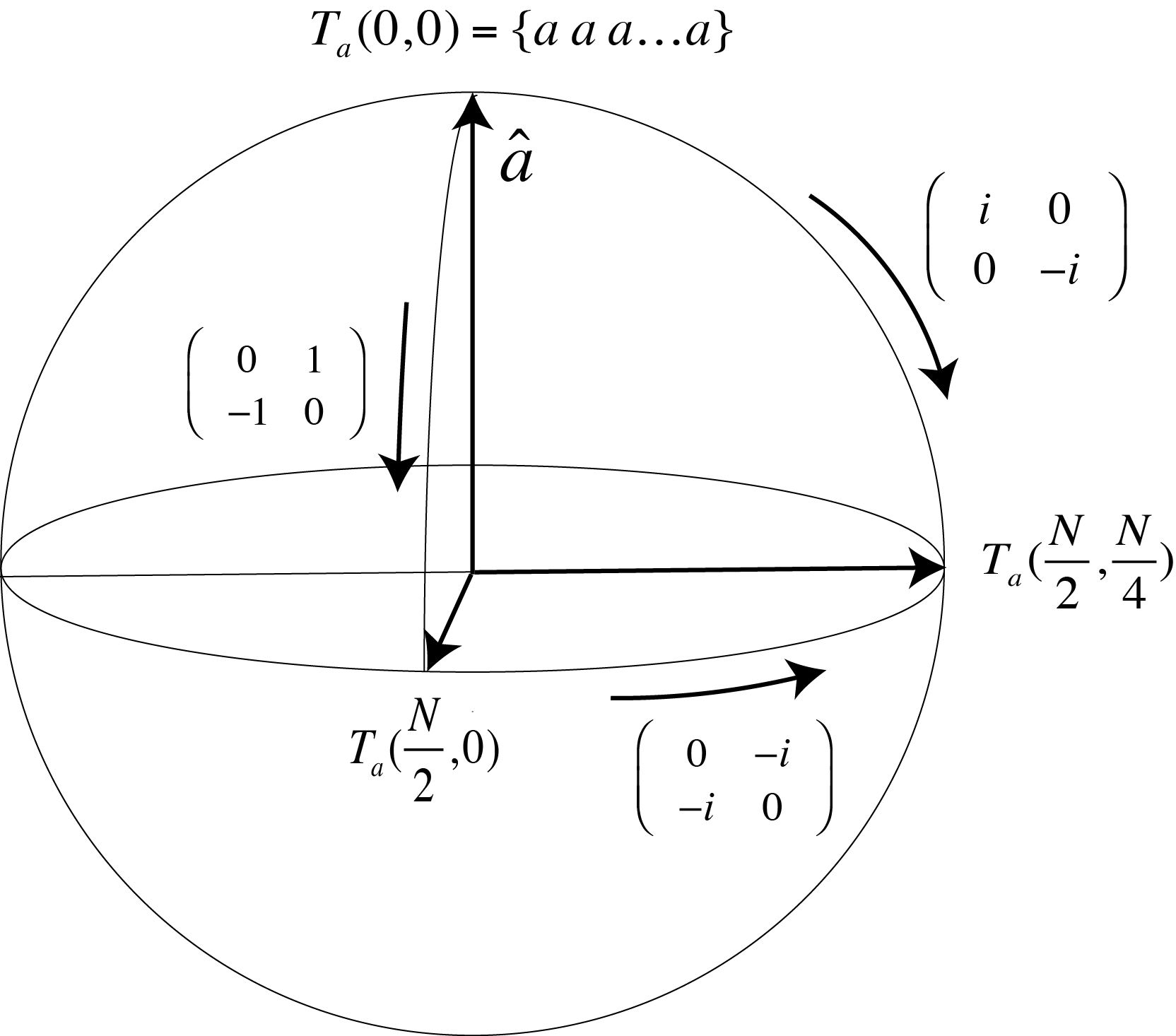}
\caption{\emph{Bit strings represented as points on a discretisation of the sphere, transformed into one another using quaternionic bit-string operators. The precise form of the representation is described in Section \ref{mapping}.}}
\label{quatern}
\end{figure}
%..........................................................................................................................................

\subsection{Number Theorems}
\label{number}
Central to this paper is Niven's theorem:
\begin{lemma}
\label{niven}
 Let $\phi/2\pi \in \mathbb{Q}$. Then $\cos \phi \notin \mathbb{Q}$ except when $\cos \phi =0, \pm \frac{1}{2}, \pm 1$. \cite{Niven, Jahnel:2005}
\end{lemma}
\begin{proof} 
Assume that $2\cos \phi = a/b$ where $a, b \in \mathbb{Z}, b \ne 0$ have no common factors.  Since $2\cos 2\phi = (2 \cos \phi)^2-2 \nonumber$ then 
\be
2\cos 2\phi = \frac{a^2-2b^2}{b^2} \nonumber
\ee
Now $a^2-2b^2$ and $b^2$ have no common factors, since if $p$ were a prime number dividing both, then $p|b^2 \implies p|b$ and $p|(a^2-2b^2) \implies p|a$, a contradiction. Hence if $b \ne \pm1$, then the denominators in $2 \cos \phi, 2 \cos 2\phi, 2 \cos 4\phi, 2 \cos 8\phi \dots$ get bigger without limit. On the other hand, if $\phi/\pi=m/n$ where $m, n \in \mathbb{Z}$ have no common factors, then the sequence $(2\cos 2^k \phi)_{k \in \mathbb{N}}$ admits at most $n$ values. We have a contradiction. Hence $b=\pm 1$ and $\cos \phi =0, \pm\frac{1}{2}, \pm1$. 
\end{proof}
We now define the three sets of angles:
\begin{align}
\label{X}
X_1=&\{0\le \phi \le 2\pi: \frac{\phi}{2\pi}= \frac{n}{N}, \ 1\le n \le N\}  \nonumber \\
X_2=&\{0 < \theta < \pi: \cos \theta = 1-\frac{2m-1}{N/2}, \ 1 \le m \le N/2 \}\nonumber \\
X_3=&\{0 < \theta < \pi: \sin \theta = 1-\frac{2m-1}{N/2}, \ 1 \le m \le N/2 \}
\end{align}
%%\cup& \{\pi < \phi < 2\pi: \sin \phi = \frac{-2-3N+4n}{N}, \ N/2 \le n \le N \}
The $N$ angles that belong to $X_1$ are distributed uniformly in radians whilst the $N/2$ angles that belong to each of $X_2$ and $X_3$ are distributed uniformly in the cosine and sine of angle respectively. By construction, $X_2$ and $X_3$ exclude the angles $\theta= 0, \pi/2, \pi$. Since $N$ is a power of 2, $X_2$ and $X_3$ also exclude $\theta=\pi/3$. Hence $X_2$ and $X_3$ exclude all the exceptions in Niven's theorem. Note that in $X_2$, $\cos^2 \theta/2=1-\frac{2m-1}{N}$, which varies between $1-1/N$ and $1/N$ as $m$ varies between $1$ and $N/2$.

We now have:
\begin{theorem}
\label{123}
The sets $X_1$, $X_2$ and $X_3$ are mutually disjoint
\end{theorem}
\begin{proof}
By Lemma \ref{niven}, $X_1$ and $X_2$ are disjoint. To show that $X_1$ and $X_3$ are disjoint then consider Lemma \ref{niven} but with $2 \sin \phi =a/b$. Because $2\cos 2\phi=2-(2\sin \phi)^2=(2b^2-a^2)/b^2$, the proof in Lemma \ref{niven} goes through with $\sin$ replacing $\cos$, with no essential  changes. To show that $X_2$ and $X_3$ are mutually disjoint, suppose to the contrary, that $\theta_0$ belongs to both $X_2$ and $X_3$. Since $\cos^2 \theta_0 + \sin^2 \theta_0=1$, then based on (\ref{X}), there must exist integers $m$ and $m'$ such that 
\be
(2m-1-N')^2+(2m'-1-N')^2=N'^2
\ee
where $N'=N/2$. Since $M\ge2$, $N'=2^{M-1}$ is even, and the right hand side is divisible by 4. Squaring out the brackets on the left hand side, it is easily seen that the left hand side, whilst divisible by 2, cannot be divisible by 4. Hence the supposition that $\theta_0$ exists is false. 
\end{proof}

\subsection{Mapping Bit Strings onto the Sphere}
\label{mapping}

Let $(p_x, p_y, p_z)$ denote the set of points on the unit sphere $\mathbb S^2$ corresponding to an arbitrary orthonormal triple of vectors $(\hat x, \hat y, \hat z )$ in $\mathbb R^3$. Let $\mathcal C_x$ denote a colatitude/longitude coordinate system with the north pole ($\theta=0$) at $p_x$, and $p_y$ and $p_z$ lie on the equatorial circle at $\phi=0$ and $\phi=\pi/2$ respectively. Let $\mathcal C_y$ and $\mathcal C_z$ denote corresponding coordinate systems with the north pole at $p_y$ and $p_z$ respectively. We define bit strings $S_x(\theta, \phi)$, $S_y(\theta, \phi)$ and $S_z(\theta, \phi)$ as follows:
\begin{align}
\label{sz}
\text{ for  } (\theta, \phi) \in \mathcal C_x \ \ \ \ S_x(\theta, \phi)&=T_a(2m-1,n);\nonumber \\
\text{ for  }(\theta, \phi) \in \mathcal C_y \ \ \ \ S_y(\theta, \phi)&=T_b(2m-1,n);\nonumber \\
\text{ for  }(\theta, \phi) \in \mathcal C_z \ \ \ \ S_z(\theta, \phi)&=T_c(2m-1,n).
\end{align}
where $\cos^2 \theta/2=1-(2m-1)/N$ and $\phi/2\pi=n/N$, $1 \le m \le N/2$, $1\le n\le N$. In the first of these, the $\hat a$ axis in Fig \ref{quatern} corresponds to the $\hat x$ axis; in the second, it corresponds to the $\hat y$ axis; in the third it corresponds to the $\hat z$ axis.  Hence, $S_x$, $S_y$ and $S_z$ are defined at those points $p$ of $\mathbb S^2$ corresponding to 
\begin{align}
\mathcal F_x &=\{p(\theta, \phi): (\theta, \phi) \in \mathcal C_x, \ \theta \in X_2, \phi \in X_1\} \nonumber \\
\mathcal F_y &=\{p(\theta, \phi): (\theta, \phi) \in \mathcal C_y, \ \theta \in X_2, \phi \in X_1\}\nonumber \\
\mathcal F_z &=\{p(\theta, \phi): (\theta, \phi) \in \mathcal C_z, \ \theta \in X_2, \phi \in X_1\}
\end{align}
respectively. We now have the central result:
 \begin{theorem}
 \label{disjoint}
 The skeletons $(\mathcal F_x, \mathcal F_y, \mathcal F_z)$ are pairwise disjoint. 
 \end{theorem}
 \begin{proof}
Start by considering the intersection of $\mathcal F_x$ and $\mathcal F_z$. Let $p(\theta, \phi)$ denote a point in  $\mathcal F_z$, so that the angular distance of $p$ from $p_z$ is $\theta$ and the angular distance of $p$ from $p_x$ is $\theta'$. Consider the spherical triangle $\triangle(p,p_x,p_z)$ - see Fig \ref{sphtri}. By the cosine rule for spherical triangles (where the angular distance between $p_x$ and $p_z$ is $\pi/2$ radians):
\be
\label{cos}
\cos \theta' = \sin \theta \cos \phi
\ee
Squaring (\ref{cos}) and using Lemma \ref{niven}, then $2 \phi$ can only equal $0, \pi/2, \pi, 3\pi/2, 2\pi$ radians. The angles $\phi=0, \pi/2, \pi \ldots$ are immediately ruled out by Theorem \ref{123}. If $\phi= \pi/4, 3\pi/4 \ldots$ then, squaring (\ref{cos}), 
\be
2(2m-1-N')^2+(2m'-1-N')^2=N'^2
\ee
where again, $N'=N/2$. Similar to part of the proof to Theorem \ref{123}, the right hand side is divisible by 4, but the left hand side cannot be. Hence the intersection between $\mathcal F_x$ and $\mathcal F_z$ is the empty set. Similar arguments shown that the intersection between $\mathcal F_y$ and $\mathcal F_z$ and between $\mathcal F_x$ and $\mathcal F_y$ is also the empty set.
\end{proof}
%..........................................................................................................................................
\begin{figure}
\centering
\includegraphics[scale=0.3]{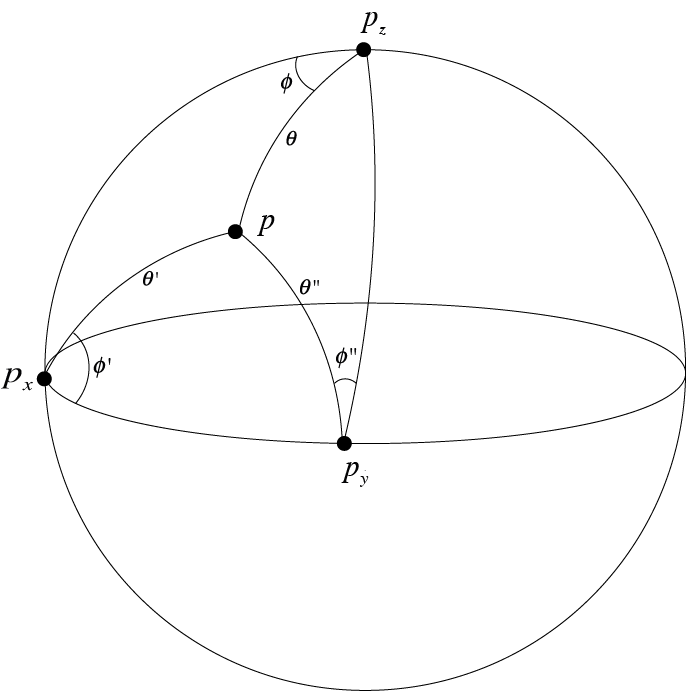}
\caption{\emph{Analysis of the properties of spherical triangles formed from the points $p$, $p_z$, $p_x$ and $p_y$ is essential to explore the consequences of this finite representation of the Bloch sphere.}}
\label{sphtri}
\end{figure}
%..........................................................................................................................................
In fact a more general result is provable. Let $p$ and $p_O$ denote points on $\mathcal F_x$ with coordinates $(\theta, \phi)$ and $(\Theta, \Phi)$ respectively. Consider the spherical triangle $\triangle(p, p_O, p_x)$ and suppose the angular distance between $p$ and $p_O$ equals $\theta'$. Then from the cosine rule for spherical triangles:
\be
\label{cossq}
\cos \theta'= \cos\theta \cos \Theta+\sin\theta \sin\Theta \cos (\phi-\Phi)
\ee
Rearranging and squaring (\ref{cossq}):
\be
\label{cossq2}
\cos^2 (\phi -\Phi )= 
\frac{(\cos \theta' - \cos \theta \cos \Theta)^2}{(1-\cos^2 \theta)(1-\cos^2 \Theta)}
\ee
Now consider the skeleton $\mathcal F_O$ associated with $p_O$. If $p \in \mathcal F_O$ then both the numerator and the denominator of the right-hand side of (\ref{cossq2}) would be rational. Hence, by Lemma \ref{niven}, $\phi-\Phi$ and hence $\phi$ must typically be irrational, whence $p$ cannot simultaneously belong to $\mathcal F_x$ and $\mathcal F_O$.  

Theorem \ref{disjoint} states that if we take a point $p$ that lies, say, on $\mathcal F_x$, then it cannot lie on either $\mathcal F_y$ and $\mathcal F_z$. As discussed in Section \ref{physics}, this provides the number-theoretic basis for complementarity and contextuality in the present theory.

\subsection{Relation to Complex Hilbert Vectors}
Using the quaternionic representations (\ref{quat}), the familiar Pauli spin matrices can be written in (bit-string) operator form
\be
\label{Pauli}
\sigma_x= \begin{pmatrix}
i & 0 \\
0 & i
\end{pmatrix}
\mathbf{i}_1;\ \ 
\sigma_y= \begin{pmatrix}
i & 0 \\
0 & i
\end{pmatrix}
\mathbf{i}_2;\ \ 
\sigma_z= \begin{pmatrix}
i & 0 \\
0 & i
\end{pmatrix}
\mathbf{i}_3.
\ee
With $(|\uparrow\rangle, |\downarrow\rangle)$ the eigenvectors of $\sigma_z$, $(|\rightarrow\rangle, |\leftarrow\rangle)$ the eigenvectors of $\sigma_x$ and $(|\odot\rangle, |\otimes\rangle)$ the eigenvectors of $\sigma_y$ (where the symbols reflect directions in space), we can identify the symbolic labels \cite{Schwinger}  $c$ and $\cancel c$ with $\uparrow$ and $\downarrow$, $a$ and $\cancel a$ with $\rightarrow$ and $\leftarrow$, and $b$ and $\cancel b$ with $\odot$ and $\otimes$, respectively (c.f. (\ref{sz})). 

Using this, families of $N-$element bit strings can be associated with complex Hilbert vectors as follows:
\begin{align}
\label{1qubit}
\mathscr S_x(\theta', \phi')& \equiv \cos \frac{\theta'}{2} |\rightarrow\rangle + e^{i \phi'} \sin \frac{\theta'}{2} |\leftarrow\rangle  \nonumber \\
\mathscr S_y(\theta'', \phi'') & \equiv \cos \frac{\theta''}{2} |\odot\rangle + e^{i \phi''} \sin \frac{\theta''}{2} |\otimes\rangle  \nonumber \\
\mathscr S_z(\theta, \phi) & \equiv \cos \frac{\theta}{2} |\uparrow\rangle + e^{i \phi} \sin \frac{\theta}{2} |\downarrow \rangle 
\end{align}
Here $\mathscr S_x$ (similarly $\mathscr S_y$ and $\mathscr S_z$) denotes an equivalence class of bit strings. In particular, $S_x(\theta, \phi)$ and $S'_x(\theta, \phi)$ belong to $\mathscr S_x(\theta, \phi)$ if there exists a `global' permutation $P$, i.e. one independent of $\theta$ and $\phi$, such that $S'_x(\theta, \phi)=P \circ S_x(\theta, \phi)$. In the physical model discussed in Section \ref{physics}, this captures the notion that two Hilbert vectors are physically equivalent if they differ by a global phase factor.  Other properties can be noted:
\begin{itemize}
\item The square of the amplitude $\cos \theta/2$ (etc) in the complex Hilbert vector is equal to frequency of occurrence $1-(2m-1)/N$ of the symbolic label $a$ in the associated bit string. 
\item From a multiplicative point of view, the permutation operator $\zeta$ is an $N$th root of unity and can be associated with the complex number $e^{2 \pi i /N}$. 
\item According to Theorem \ref{disjoint}, a point $p \in \mathbb S^2$ associated with one of the three  Hilbert vectors in (\ref{1qubit}) cannot be associated with either of the other two Hilbert vectors. We relate this to the notions of complementarity and contextuality in Section \ref{physics} below. 
\end{itemize}

\subsection{The Uncertainty Principle}
\label{hilbert}

If we simply substitute the digit $1$ for $a$ and the digit $-1$ for $\cancel a$, then the mean value $\bar S_x(\theta, \phi)$ of $S_x(\theta, \phi)$ over all $N$ elements of the bit string is equal to $\cos \theta$ (similarly for $\bar S_y$ and $\bar S_z$). Moreover, the standard deviation $\Delta S_x$ of $S_x(\theta, \phi)$ over all elements of the bit string is equal to $\sin \theta $ (similarly for $\Delta S_y$ and $\Delta S_z$). 

A quantum-like `uncertainty principle' arises from these properties, together with a simple geometric property of the triangle $\triangle(p\;p_xp_y)$ (Fig \ref{sphtri}). Now if $\theta'=\pi/2$ in Fig \ref{sphtri} then $\triangle(p\;p_xp_y)$ would contain two internal right angles and application of the cosine rule immediately gives $\theta''=\phi'$. However, if $\theta'$ is either greater or smaller than $\pi/2$ then $\theta'' > \phi'$. Hence we can write
\be
\label{ineq}
|\sin \phi'| \le |\sin \theta''|
\ee
Now using the cosine rule for the spherical triangle $\triangle(p\;p_x p_z)$ we have 
\be
\cos \theta= \sin \theta' \sin \phi'
\ee
Taking absolute values and substituting in (\ref{ineq}) we have
\be
\label{uncert1}
|\sin \theta'| |\sin \theta''| \ge |\cos \theta|
\ee
i.e.
\be
\label{uncert2}
\Delta S_x\Delta S_y \ge |\bar S_z|
\ee
If instead of $a=1$ we had written $a=\hbar/2$ then (\ref{uncert2}) would be replaced by:
\be
\label{uncert3}
\Delta S_x\Delta S_y \ge \frac{\hbar}{2} |\bar S_z|
\ee
Suppose $p \in \mathcal F_z$, so that the mean value $\bar S_z(\theta, \phi)$ is defined on $\mathcal F_z$, then by Theorem \ref{disjoint}, the standard deviations $\Delta S_x(\theta', \phi')$ and $\Delta S_y(\theta'', \phi'')$ are undefined on $\mathcal F_z$. This relates to the notion of contextuality discussed above (and in Section \ref{SG} below) and relates to the notion that when states are defined by Hilbert vectors, in quantum theory these states cannot be said to have simultaneously well-defined properties relative to non-commuting observables. On the other hand, for large enough $N$, any neighbourhood $\mathcal N_p$ of $p$, no matter how small, will contain points, some of which lie in $\mathcal F_z$, some in $\mathcal F_x$ and some in $\mathcal F_y$. Hence $\mathcal N_p$ contains points where collectively all of $\bar S_z$, $\Delta S_x$ and $\Delta S_y$ are defined and hence where (\ref{uncert1}) and (\ref{uncert2}) (and hence (\ref{uncert3})) are defined. This relates to the fact that in quantum theory, the uncertainty principle describes a property associated ensembles of experiments defined over non-commuting operators: the first ensemble estimating $\Delta S_x$, the second ensemble estimating $\Delta S_y$ and the third ensemble estimating $\bar S_z$. According to the analysis above, such distinct ensemble experiments must collectively satisfy the uncertainty relation (\ref{uncert3}). 

That is to say, in this bit-string representation of complex Hilbert vectors, the famous uncertainty relationships of quantum theory arise simply as consequences of the geometry of spherical triangles. 

\subsection{Relation to Meyer's Construction}
\label{meyer}
%..........................................................................................................................................
\begin{figure}
\centering
\includegraphics[scale=0.3]{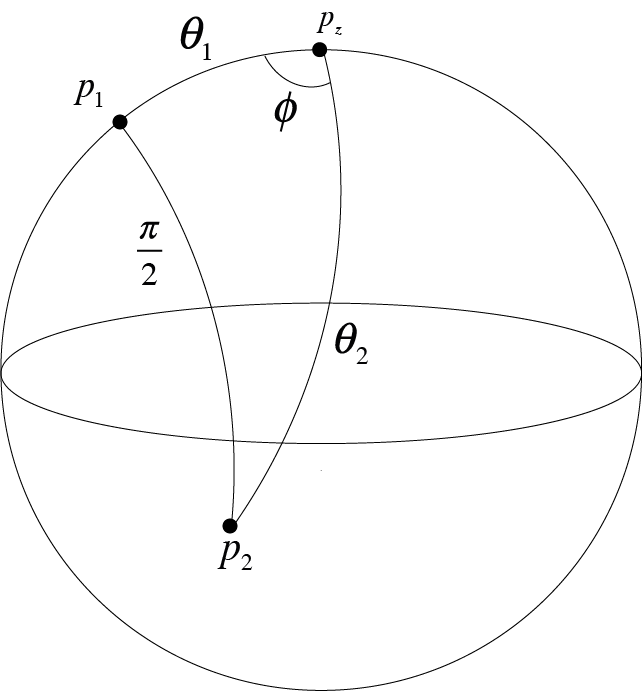}
\caption{\emph{A spherical triangle used to show that a finite skeleton such as $\mathcal F_z$ cannot contain triples of orthonormal points.}}
\label{sphtri2}
\end{figure}
%..........................................................................................................................................

This work can be related to a programme started by Meyer \cite{Meyer:1999} who showed it was possible to nullify the Kochen-Specker theorem by considering a rational set of triples of 'colourable' points which are dense in $\mathbb S^2$, each triple corresponding to an orthonormal vector triad in $\mathbb R^3$. However, there is a crucial difference with this and Meyer's construction:
\begin{theorem}
 No orthogonal triples of points exist in a finite skeleton such as $\mathcal F_z$. 
 \end{theorem}
 \begin{proof}
As above, the proof is based on the number theorems in Section \ref{number} again applied to the cosine rule for spherical triangles. Let us suppose that the orthonormal triple $(p_1, p_2, p_3)$ exists in $\mathcal F_z$, and hence in particular where the angular distance between $p_1$ and $p_2$ is $\pi/2$ radians (see Fig \ref{sphtri2}). Let $\phi$ denote the angle at $p_z$ between the longitudes which contain $p_1$ and $p_2$ respectively. Then, from the cosine rule applied to the triangle $\triangle(p_1,p_2,p_z)$ (see Fig \ref{sphtri2}):
\be
 \label{triple}
 0=\cos \theta_1 \cos \theta_2+ \sin \theta_1 \sin \theta_2 \cos \phi
 \ee
If we square this equation, then necessarily $\cos 2\phi$ must be rational. Hence $\phi=0, \pi/4, \pi/2 \ldots$ by Lemma \ref{niven}. 

\begin{itemize}
\item If $\phi=0$ then $\theta_{2}=\theta_1+\pi/2$ and hence $\cos \theta_{1}=-\sin \theta_2$. This has been ruled out by Theorem \ref{123}. 
 
\item If $\phi=\pi/2$, then we have $0=\cos \theta_1 \cos \theta_{2}$ and either $\theta_1$ or $\theta_{2}=\pi/2$. However, from (\ref{X}), $X_2$ does not contain such an angle. 

\item The remaining possibility is $\phi=\pi/4$. Squaring (\ref{triple}) with $\cos^2 \phi =1/2$ we have
\be
\label{trip}
\cos^2 \theta_1 \cos^2 \theta_2 + \cos^2 \theta_1 + \cos^2 \theta_2=1
\ee
From (\ref{X}), we now substitute $1-\frac{2m_1-1}{N'}$ for $\cos\theta_1$ and $1-\frac{2m_2-1}{N'}$ for $\cos\theta_2$ giving  
 \be
 (2m_1-1-N')^2(2m_2-1-N')^2+N'^2(2n_1-1-N')+N'^2(2n_2-1-N)=N'^4
 \ee
Since $N'=N/2$ is even, the right-hand side is divisible by 16. Expanding out the terms, it is clearly seen that the left-hand side is not divisible by 16. Hence $\phi \ne \pi/4$. 
\end{itemize}
\end{proof}
Hence there are no orthonormal triples. In this sense, the present construction does not nullify the Kochen-Specker theorem. Indeed, as discussed, the construction is explicitly contextual in character: if a point $p$ lies in $\mathcal F_z$ (and so has well defined symbolic attributes relative to the $z$ direction) it does not lie in $\mathcal F_x$ (and so does not have symbolic attributes relative to the $x$ direction). On the other hand, as discussed below, the present construction does largely nullify the Bell Theorem, whilst the Meyer construction does not. 
\subsection{Fine Tuning}
The construction above appears to be rather finely tuned; for large $N$, seemingly tiny perturbations can take a point off the discretised skeleton. However, such a conclusion depends on a choice of metric. In the next Section we discuss an interpretation of this discretisation, where points on the discretised Bloch sphere describe ensembles of trajectories on a fractal subset of state space. As discussed below, using a $p$-adic metric \cite{Katok} which respects the primacy of this fractal, any perturbation which takes a state off the fractal subset is a large amplitude perturbation by definition, even though it may appear small from a Euclidean perspective. In this sense, the discretisation of the sphere described above is not fine tuned. 

\section{The Physics of Finite Single Qubit States}
\label{physics}
\subsection{Finite Hilbert Vectors and Invariant Set Theory}
\label{invariant}

This work was largely motivated by the elementary observation that the form
\be
\label{L}
\frac{\partial \rho}{\partial t} =  \{H,\rho\}
\ee
of the Liouville equation
in classical physics is simply too close to the von Neurmann form 
\be
\label{S}
i\hbar\frac{\partial \rho}{\partial t} =  [H,\rho]
\ee
of the Schr\"{o}dinger equation to be a coincidence. Hence, just as the linearity of the Liouville equation - a consequence of conservation of probability - says nothing about the nonlinearity or determinism of the dynamical system that generates the probability density field, we should similarly not infer from the linearity of the Schr\"{o}dinger equation that fundamental physics can be neither nonlinear nor deterministic. However, in seeking a nonlinear deterministic underpinning for the Schr\"{o}dinger equation, it is clearly vital to understand why the constants $\hbar$ and $i$ appear in (\ref{S}) and not (\ref{L}) - and  why (\ref{S}) and not (\ref{L}) is formulated in the language of Hilbert Space states and operators. The constructive representation of complex Hilbert states discussed above provides a novel perspective on these issues. 

In classical theory, the only state-space trajectory that matters is the one with least action - the others play no role. In quantum theory this is not so as the Feynman path integral approach indicates explicitly. The notion of the importance of extended state-space structure is manifest in other approaches to quantum theory. For example in Bohmian theory the classical Hamilton-Jacobi equation is supplemented with the gradient of the quantum potential, the latter being a function on configuration space. All this is consistent with the fact that the essential constant, $\hbar$, has dimensions of momentum times position, i.e. of phase space. 

The essential idea developed in this Section is that processes occurring in space-time are associated with the extended state-space structure of a primal fractal-like geometry $I_U$ in state space. In this way, the nonlinear deterministic dynamics that underpins the Schr\"{o}dinger equation is defined by the fractal structure of state-space trajectories on $I_U$ (see Fig \ref{helix}). In particular, suppose that at some $I$th level of fractal iteration, the trajectory of a physical system in its state space is represented as a simple 1D curve. Suppose (at this level of approximation) this trajectory bifurcates into two, each bifurcated trajectory reaching and becoming quasi-stationary in distinct parts of state space labelled $a$ and $\cancel a$. At this level of approximation, the dynamical evolution of the system is reminiscent of Everettian branching. However, in the model developed here, such branching arises simply because we are looking at the trajectories at an inadequate level of fractal iteration. Suppose at the $I+1$th level of fractal iteration, the $I$th-iterate trajectory is found to comprise a set of $N$ trajectories, drawn as a helix in Fig \ref{helix}.  At this $I+1$th iterate, the trajectory segments do not bifurcate as they approach $a$ and $\cancel a$, but instead diverge (consistent with a simple deterministic instability). As shown in Fig \ref{helix}, each trajectory in the helix can be labelled by the state-space cluster to which it evolves - that is to say, each trajectory within the helix at the $I+1$th iteration can be given the symbolic label $a$ or $\cancel a$. 
%..........................................................................................................................................
\begin{figure}
\centering
\includegraphics[scale=0.3]{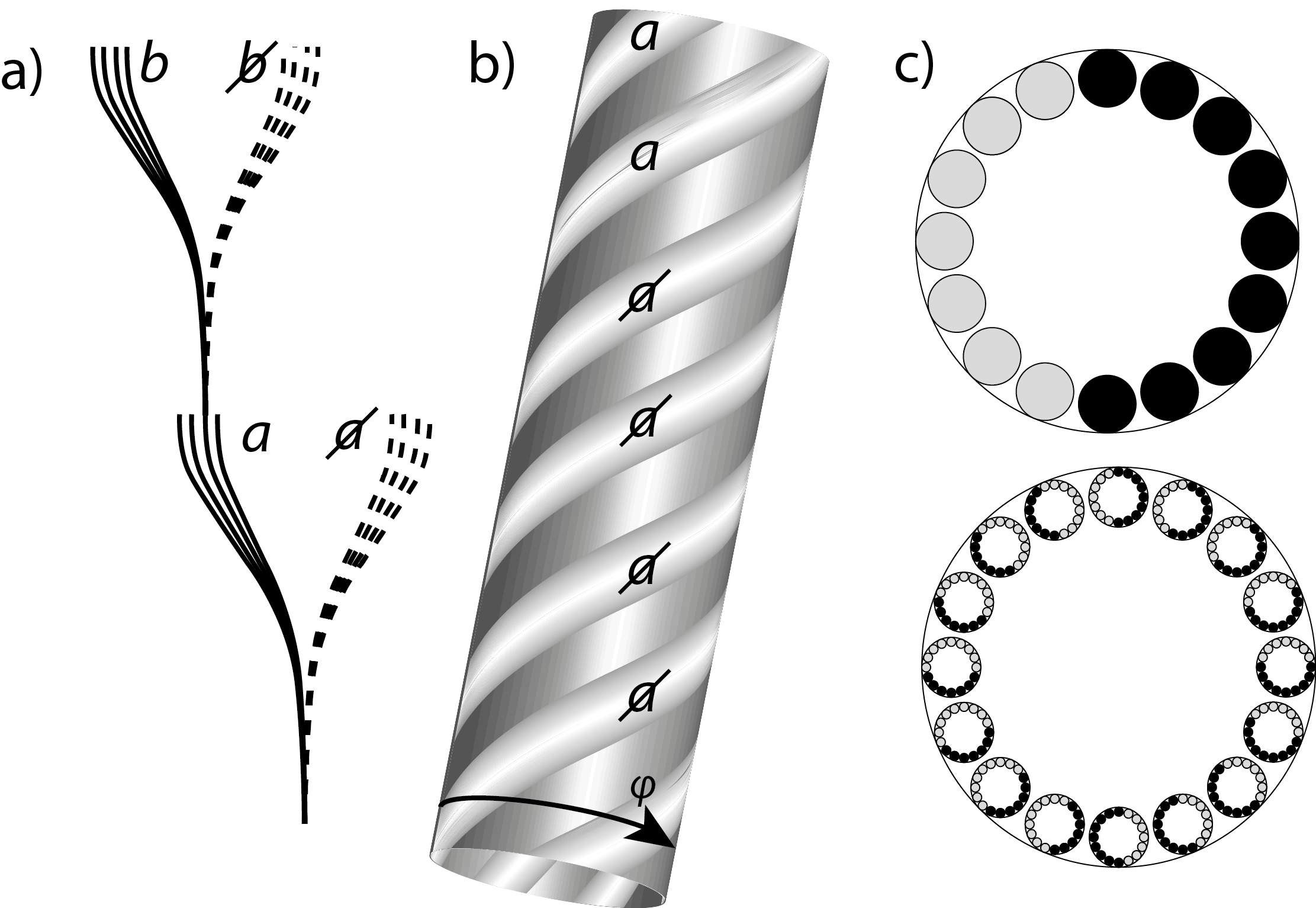}
\caption{\emph{A schematic of the local state-space structure of the invariant set. a) An ensemble of trajectories decoheres into two distinct regions of state space labelled $a$ and $\cancel a$. Under a second phase of decoherence, this trajectory, itself comprising a further ensemble, decoheres into two further distinct regions of state space labelled $b$ and $\cancel b$. b) The fractal structure of trajectories is such that under magnification in state space, a trajectory segment is found to comprise a helix of $N$ trajectories at the next fractal iterate. c) Top: a cross section through the helix of trajectories comprises $N=16$ disks coloured black or grey according to whether that trajectory evolves to the $a$ regime or the $\cancel a$ regime. Bottom: each of these $N$ disks itself comprises $N$ further disks coloured black or grey according to whether each trajectory evolves to the $b$ or $\cancel b$. A fractal set of disks is homeomorphic to the set of $N$-adic integers.}}
\label{helix}
\end{figure}
%..........................................................................................................................................

Continuing in this way, it is supposed that each trajectory of the helix at the $I+1$th iteration itself comprises a helix of $N$ trajectories at the $I+2$th iteration. Each of these $I+2$th interate trajectories evolve to the distinct regions $b$ and $\cancel b$ of state space (Fig \ref{helix}a), and so on. A cross section of the fractal helix is a topological representation of the set of $N=2^M$-adic (and hence 2-adic) integers \cite{Katok}. From the perspective of the $2$-adic metric, a point not lying on the iterated fractal is necessarily distant from a point on the set (no matter how close the two points are from a Euclidean perspective). The primacy of the $N$-adic distance in state-space is a key reason why invariant set theory is not a simple classical theory. $p$-adic analysis has already been applied to formulate alternative approaches to quantum theory \cite{Khrennikov} \cite{Volovich}. As these authors also acknowledge, such an approach is indicative of the primal role that fractal geometry must be playing in fundamental physics. In the present model, such geometry primarily arises in state space, rather than space-time. 

As mentioned, the labels $a$ and $\cancel a $ refer to specific subsets of some two dimensional cross section of (cosmological) state space, associated with the clustering of trajectories. Here we assume that when two trajectories in Fig \ref{helix} have diverged sufficiently from one another, the system's environment also evolves differently (c.f. the butterfly effect \cite{Palmer:2014b}). That is to say, it is assumed that the state-space clustering associated with $a$ and $\cancel a$ extend into multiple dimensions of state space, so that $a$ and $\cancel a$ are associated with distinct macroscopic outcomes (e.g. particle detected by macroscopic detector, or particle not detected by macroscopic detector). We do not attempt to describe this divergence and clustering process in closed mathematical form in this paper, though note that this may best be achieved using $p$-adic dynamical systems analysis \cite{WoodcockSmart}. Also, the extent to which the regions $a$ and $\cancel a$ involve sufficiently many of the environmental degrees of freedom as to become gravitationally distinct \cite{Diosi:1989} \cite{Penrose:1994}, a matter of considerable importance at a foundational level, is beyond the scope of this paper, except to say that the process of trajectory clustering into distinct regimes could conceivably provide an emergent description of gravity.  We leave these issues to another paper. 

In quantum physics we typically consider systems that have been prepared in some initial state, are then subject to various transformations, following which certain measurements are made. In this framework, preparation can be described by first considering an $I$th trajectory which has already evolved to an $a$ cluster. As described, this trajectory comprises a helix of $I+1$th iterate trajectories, $N-m$ of which evolve to some cluster $b$ and $m$ of which evolve to some cluster $\cancel b$. That is to say, all $I+1$th iterate trajectories can be labelled $a$ (the preparation eigenstate) and $N-m$ of them can be labelled $b$ and $m$ of them labelled $\cancel b$ (the measurement eigenstates). 

As discussed in Section \ref{finite}, the $N$-element bit strings such as $S_z(\theta, \phi)$ etc have quaternionic structure. That is to say, the precise order in which the helical trajectories are linked to one cluster or the other, as one rotates about the helix, encodes a measurement orientation $(\theta, \phi)$ in physical space, relative to the north pole preparation direction. This encodes the concept, outlined above, that an element of reality in space-time is encoded in the geometry of an extended region of state space. According to the analysis in Section \ref{finite}, a proportion $\cos^2 \theta/2 =1-\frac{2m-1}{N}$ of trajectories corresponding to the point $(\theta, \phi)$ are labelled with the $b$ label. In this way, the geometry of $I_U$ encodes the statistics of quantum measurement. If instead of a single system, we consider an ensemble of systems all prepared in (nominally but not precisely) the same way, then the ensemble of $a$ and $b$ orientations will lie in some small finite neighbourhoods $\mathcal N_a$ and $\mathcal N_b$, whose size is inversely proportional to the finite precision of the preparation and measurement devices. 

We now discuss Theorem \ref{disjoint} from a physical perspective. In quantum theory it is frequently said that if a particle's spin is measured relative to some direction $x$ it cannot be `simultaneously' measured relative to some other direction. This notion of a simultaneous measurement can be equated to a counterfactual measurement. Theorem \ref{disjoint} indicates that such counterfactual measurements are undefined: suppose a particle's spin was measured in the $x$ direction, then according to Theorem \ref{disjoint} it is not the case that the particle's spin is well defined relative to the $y$ or $z$ directions. In the context of $I_U$, one is saying that if a measurement relative to the $x$ direction lies on $I_U$, then a measurement \emph{on that same particle} relative to the $y$ or $z$ direction cannot lie on $I_U$. This is consistent with the notion of measurement contextuality in quantum theory (and is in turn consistent with the fact discussed in Section \ref{meyer} that invariant set theory respects the Kochen-Specker theorem). 

These considerations must be viewed within the bigger picture where $I_U$ describes a global but compact fractal geometry of cosmological states \cite{Palmer:2009a} \cite{Palmer:2014}, analogous with the fractal attractors of nonlinear dynamical systems theory. $I_U$ is referred to as the (cosmological) invariant set, and embodies the Bohmian holistic notion of an undivided universe \cite{BohmHiley}. Here it is assumed that the universe evolves through a very large number of aeons so that in an ensemble of $N$ $I$th-iterate trajectories on $I_U$, one trajectory corresponds to the current aeon, and the others correspond to past or future aeons (see also Section \ref{Bell}). In this sense, ordering of points corresponding to the passage of time along trajectories corresponds to Bohmian explicate order, whilst ordering of points corresponding to ($N$-adic) distance transverse to trajectories corresponds to Bohmian implicate order. Henceforth we refer to this model of quantum physics as `invariant set theory'. If the number of distinct aeons is finite, $I_U$ is ultimately periodic in structure. 

Since the experimenter is part of the physical universe, the neuronal processes which determine her choice of measurement setting are also described by the geometry of $I_U$. Now if $N$ were a small number (e.g. the minimal value $N=4$), then the number of orientations available to the experimenter would be smaller than the number of orientations conceivable by her neuronal processes. This would contradict the fact that for all practical purposes the experimenter can orient her measuring apparatus in any way she chooses, and experimental results are particular to and consistent with that particular orientation.  

No such contradiction arises if $N$ is sufficiently large, and none of the results in this paper require $N$ to have any particular upper bound. Any measurement apparatus has finite resolution, and many theories of quantum gravity suggest that, linked to the Planck scale, there is an absolute maximum resolution any measuring apparatus can have \cite{hossenfelder}. If $N$ is so great that the angular distance between any consecutive angles in $X_1$, $X_2$ or $X_3$ is smaller than this maximum resolution, then for all practical purposes the experimenter can choose any measurement orientation she likes and the experimental results will be particular to and consistent with that orientation. It may seem superficially as if there would be no observable consequences if $N$ were large but finite in this sense, rather than being simply infinite (the singular limit where the closed Hilbert Space is recovered). However, below we show that there are consequences. 

Based on this discussion we conclude this section by itemising the basic postulates of invariant set theory:
\begin{itemize}
\item States of physical reality of the universe are those which lie on a special fractal-like set $I_U$ in cosmological state space. More generally, the laws of physics at their most primitive derive from the geometrical properties of $I_U$. 
\item Only states of physical reality encompass the notion of space time, which is presumed to be relativistic. Hypothetical state space perturbations which take points on $I_U$ to points off $I_U$ are not expressible as changes to events in space-time. 
\end{itemize}
In future papers, the discretisation of the Bloch sphere will be extended to include Lorentz boosts. In this way it is hoped to show explicitly how relativistic space-time is emergent from the spinorial structure of $_U$. It can be noted that the number theorems above readily extend from trigonometric to hyperbolic functions. 

\subsection{The Sequential Stern-Gerlach Experiment}
\label{SG}

As an illustration of the construction above, and the central importance of the rational constraint on $\cos \theta$, consider a conventional sequential Stern-Gerlach experiment used to introduce students to the non-commutativity of spin measurements in quantum theory and hence contextuality of the quantum state. An ensemble of spin-1/2 particles is prepared by the first of three Stern-Gerlach apparatuses with spin oriented in the direction $\hat{\mathbf {a}}$. Some exit through the spin-up channel and enter a second Stern-Gerlach apparatus oriented in the direction $\hat{\mathbf{b}}$. The particles that are output along the spin-up channel of the second apparatus are then passed into a third Stern-Gerlach apparatus oriented in the direction $\hat{\mathbf{c}}$. Here we can think of the first apparatus as preparing the state for a measurement by the second apparatus, and the second apparatus as preparing the state for a measurement by the third apparatus. 
%..........................................................................................................................................
\begin{figure}
\centering
\includegraphics[scale=0.3]{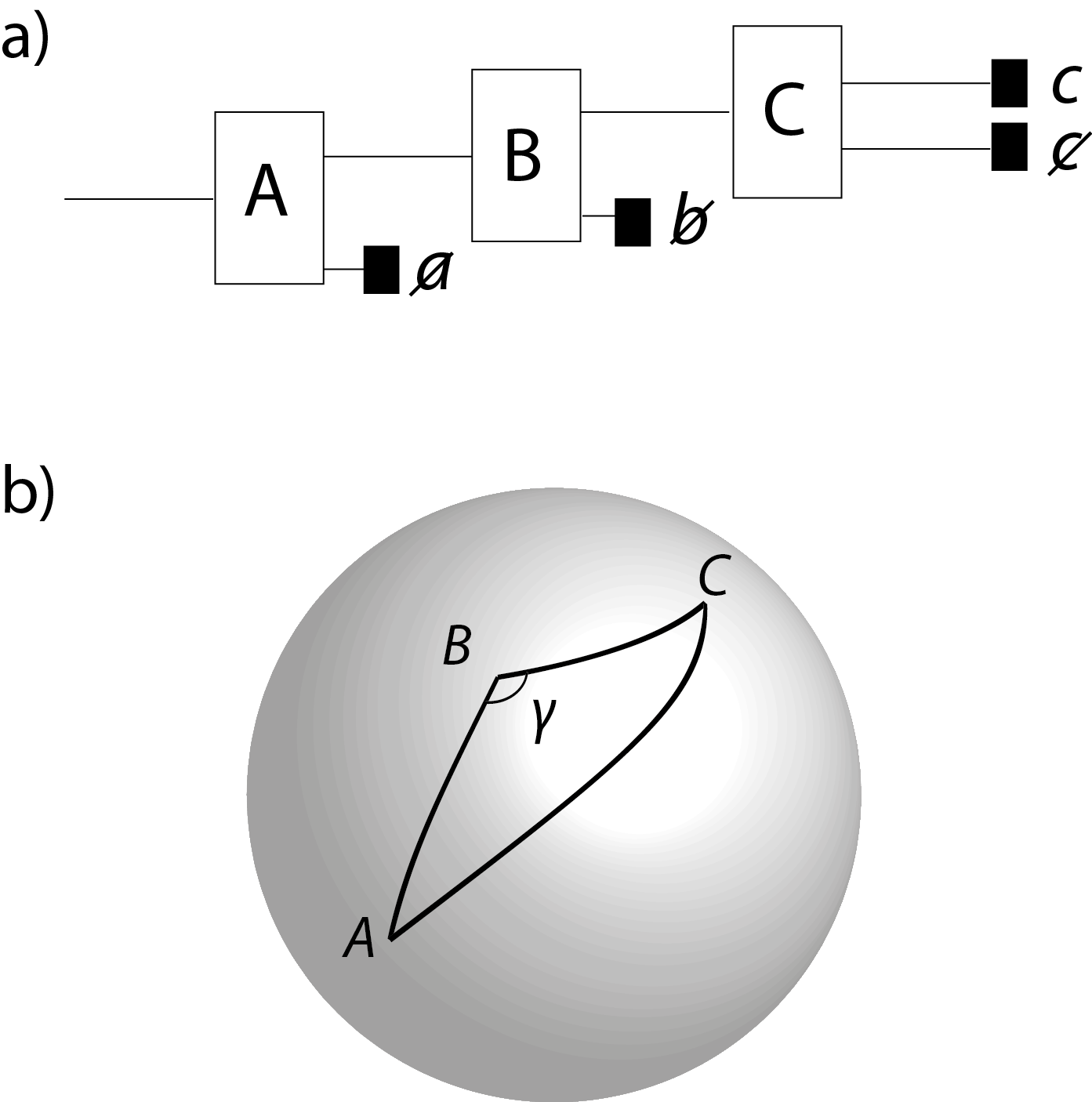}
\caption{\emph{a) A sequential Stern-Gerlach experiment where a particle is sent through three Stern-Gerlach devices. b) A, B and C shown as points on the celestial sphere. (Although to experimental accuracy, A, B and C may appear to lie on a great circle, they will not \emph{precisely}.) The incompatibility of non-commutating quantum observables in a deterministic framework is illustrated using a number-theoretic analysis of the cosine rule for spherical triangles.}}
\label{sgfig}
\end{figure}
%..........................................................................................................................................

In the discussion below, $\hat{\mathbf{a}}$, $\hat{\mathbf{b}}$ and $\hat{\mathbf{c}}$ are to be considered directions under the control of the experimenter. That is to say, they are nominal directions and can be thought of as varying slightly from one set of measurements to the next in the ensemble. Associated with the nominal directions $\hat{\mathbf{a}}$, $\hat{\mathbf{b}}$ and $\hat{\mathbf{c}}$ are three small neighbourhoods $\mathcal A$, $\mathcal B$ and $\mathcal C$ of points on $\mathbb S^2$ - the smallness of these neighbourhoods being determined by the finite precision of the Stern-Gerlach devices. For any set of measurements (on a particular particle) are three precisely defined points $A$, $B$ and $C$ drawn from each of the three neighbourhoods. In particular, although $\hat{\mathbf{a}}$, $\hat{\mathbf{b}}$ and $\hat{\mathbf{c}}$ may be coplanar to the precision that the experimenter can determine this,  a particular set $\{A, B, C\}$ of points will not lie \emph{precisely} on a single great circle and hence can be represented as the vertices of some non-degenerate spherical triangle $\triangle ABC$. Now as far as the experimenter is concerned, the nominal directions $\hat{\mathbf{a}}$, $\hat{\mathbf{b}}$ and $\hat{\mathbf{c}}$ can be set arbitrarily.  However, in invariant set theory, although the orientation associated with $A$ can be considered arbitrary (about which see below), the orientation associated with $B$ cannot: $\cos \theta_{AB}$, the cosine of the angular distance between $A$ and $B$, must be rational. Similarly, although the orientation of the great circle between $A$ and $B$ can be considered arbitrary, the orientation $\gamma$ of the great circle between $B$ and $C$, relative to the great circle between $A$ and $B$ is not arbitrary: $\gamma/2\pi$ must be rational. Finally, $\cos \theta_{BC}$, the cosine of the angular distance between $B$ and $C$ must be rational. 

Let us start by considering a simple two-stage Stern-Gerlach experiment $\hat{\mathbf{a}} \rightarrow \hat{\mathbf{b}}$. From quantum theory the frequency of spin up measurements is equal to $\cos^2 \theta/2$ where $\theta$ is the nominal angle between $\hat{\mathbf{a}}$ and $\hat{\mathbf{b}}$ . The same statistical result obtains in invariant set theory.  In invariant set theory, the experimenter does not need to take special care that $\cos \theta_{AB}$ is rational - this is beyond her control and is an automatic consequence of the postulate that the state of the universe lies on the invariant set. Hence, in an ensemble of measurements with nominal directions $\hat{\mathbf{a}}$ and $\hat{\mathbf{b}}$, the precise rational value of $\cos \theta_{AB}$ will vary from one measurement to the next. The average of $\cos \theta_{AB}$ over such a sample of measurements is equal to $\cos \theta$. 

The rationality conditions plays a crucial role when considering the order of measurements in a three-stage Stern-Gerlach device. Consider the strictly counterfactual question: although the experiment $\hat{\mathbf{a}} \rightarrow \hat{\mathbf{b}} \rightarrow \hat{\mathbf{c}}$ was performed on a particular particle, could the experiment $\hat{\mathbf{a}} \rightarrow \hat{\mathbf{c}} \rightarrow \hat{\mathbf{b}}$ have been performed on that particle? That is to say, is the order of the $\hat{\mathbf{b}}$ and $\hat{\mathbf{c}}$ Stern-Gerlach devices important? To see that it is, consider the cosine rule for spherical triangles
\be
\cos \theta_{AC}= \cos \theta_{AB}\cos \theta_{BC}+\sin \theta_{AB}\sin \theta_{BC}\cos \gamma
\ee
The right hand side is the sum of two terms. The first is rational since it is the product of two terms each of which, by construction, is rational. The second is the product of three terms the last of which, $\cos \gamma$,  is irrational ($\gamma \ne 180 ^\circ$ precisely). Since $\theta_{AB}$, $\theta_{BC}$ and $\gamma$ are independent degrees of freedom defining the triangle $\triangle ABC$, $\sin \theta_{AB}$ and $\sin \theta_{BC}$ cannot conspire with $\cos \gamma$ to make the product $\sin \theta_{AB}\sin \theta_{BC}\cos \gamma$ rational. Hence $\cos \theta_{AC}$ is the sum of a rational number and an irrational number and is therefore irrational. Therefore, if $U$ is a universe in which $\hat{\mathbf{a}} \rightarrow \hat{\mathbf{b}} \rightarrow \hat{\mathbf{c}}$ is performed on a particular particle, and therefore lies on $I_U$, the counterfactual universe $U'$ where $\hat{\mathbf{a}} \rightarrow \hat{\mathbf{c}} \rightarrow \hat{\mathbf{b}}$ is performed on the same particle cannot lie on $I_U$. In invariant set theory, the non-commutativity of spin measurement arises by number theory applied to the cosine rule for spherical triangles. From the perspective of the 2-adic metric discussed in Section \ref{invariant}, the state of the universe corresponding to the counterfactual $\hat{\mathbf{a}} \rightarrow \hat{\mathbf{c}} \rightarrow \hat{\mathbf{b}}$ experiment is distant from the state corresponding to the actual $\hat{\mathbf{a}} \rightarrow \hat{\mathbf{b}} \rightarrow \hat{\mathbf{c}}$ experiment, even though from a Euclidean perspective the angular differences between the orientations of the measuring devices may be tiny. 

We can of course perform two non-simultaneous sequential Stern-Gerlach experiments on different particles - the first with order $\hat{\mathbf{a}} \rightarrow \hat{\mathbf{b}} \rightarrow \hat{\mathbf{c}}$, the second with order $\hat{\mathbf{a}} \rightarrow \hat{\mathbf{c}} \rightarrow \hat{\mathbf{b}}$. For the first experiment, $\cos \theta_{AB}$ and $\cos \theta_{BC}$ are rational, and the angle subtended at $B$ is a rational multiple of $2 \pi$. For the second experiment, $\cos \theta_{A'C'}$ and $\cos \theta_{B'C'}$ are rational, and the angle subtended at $C'$ is a rational multiple of $2 \pi$. Here $A, A' \in \mathcal A$, $B,B' \in \mathcal B$ and $C,C' \in \mathcal C$.  

Before concluding this Section, it is worth commenting on the output probabilities for a single Stern-Gerlach device A. Here a particle leaves the source with spin oriented in some direction $(\theta_{\text{ref}}, \phi_{\text{ref}})$ represented by some point $P \in \mathbb S^2$. This direction is an expression of the relationship between the single qubit under consideration, relative to the rest of $I_U$. It can be considered part of the information embodied in a supplementary variable $\lambda$ representing that qubit, and, as far as the experimenter is concerned, is unknown (and, indeed, as discussed in Section \ref{Bell} is unknowable). For the Stern-Gerlach measurement, invariant set theory requires that $\cos \theta_{PA}$ is rational. For an ensemble of particles the unknown directions $\{P_i\}$ can be assumed random with uniform distribution on the sphere. Hence, one can expect half of the particles to emerge in the spin-up output branch of the Stern-Gerlach device, and half in the spin-down output branch. 

\section{Entangled Qubits}
\label{implications}

\subsection{Correlated bit string and Hilbert state tensor products}
\label{entangle}

In this finite representation, the state space associated with $J$ entangled qubits is simply the $J$-times Cartesian products of the finite Bloch Sphere as described above. This contrast strongly with quantum theory, and a potentially testable consequence of this is discussed below. 

Start by considering the Cartesian product $T_{ab}= T_a \times T_b$ of two bit strings $T_a=\{a_1 \; a_2 \ldots a_{N}\}$, $T_b=\{b_1 \; b_2 \ldots b_{N}\}$ where $a_i \in \{a, \cancel a\}$, $b_i \in \{b, \cancel b\}$ and
\begin{align}
\label{c1}
T_{ab}(m_1, m_2, m_3) &=\{\overbrace{ a \ a\  \ldots \  a\ \ \  a \  a \ \ldots \  a}^{N-m_1}\ \ \ \overbrace{\cancel a\ \cancel a\ \ldots \ \cancel a\ \ \cancel a\ \cancel a\ \ldots \ \cancel a}^{m_1} \}\nonumber \\
& \times \{\underbrace{b \  b \ \ldots \  \  b}_{N-m_2}\ \ \ \underbrace{\cancel b\  \cancel b\ \ldots \  \cancel b}_{m_2-m_1}\ \ \ \ 
\underbrace{\cancel b \ \cancel b \ \ldots \ \cancel b}_{m_3}\ \ \underbrace{b\ b\ \ldots \ b}_{m_1-m_3} \} \nonumber \\ 
&=T_a(m_1) \times \zeta^{m_1-m_3}T_{b}(m_2+m_3-m_1)
\end{align}
with three degrees of freedom represented by the integers $m_1, m_2, m_3$. Three cyclic permutation operators, $\zeta_1$, $\zeta_2$ and $\zeta_3$ can be defined on $T_{ab}$ without impacting on the correlation between $T_a$ and $T_b$. $\zeta_1$ is defined by
\begin{align}
\zeta^{n_1}_1 T_{ab}= \zeta^{n_1} T_a \times \zeta^{n_1} T_b
\end{align}
also represented as 
\begin{align}
\zeta^{n_1}_1 T_{ab}=&\{\overbrace{a \ a\  \ldots a \  \ \ \ a  \ a \ \ldots \ a \ \ \ \cancel a\ \cancel a\ \ldots \ \cancel a\ \ \ \  \cancel a\ \cancel a\ \ldots \ \cancel a}^{\zeta^{n_1}} \} \nonumber \\
 \times & \{\underbrace{ b \ \ b \ \ldots \  b\ \ \ \ \cancel b \ \cancel b\ \ldots \ \cancel b \ \ \ \cancel b \ \cancel b \ \ldots \ \cancel b \ \ \ \ \ b\  b\ \ldots \ b}_{\zeta^{n_1}} \} \nonumber
\end{align}
In terms of this representation, $\zeta_2$ and $\zeta_3$ are defined by
\begin{align}
\zeta_2^{n_2} T_{ab}=&\{\overbrace{a \ a\  \ldots a \  \ \ \ a  \ a \ \ldots \ a}^{\zeta^{n_2}} \ \ \ \cancel a\ \cancel a\ \ldots \ \cancel a\ \ \ \  \cancel a\ \cancel a\ \ldots \ \cancel a \} \nonumber \\
 \times & \{\underbrace{b \ b \ \ldots \ b\ \ \ \ \cancel b\ \cancel b\ \ldots \ \cancel b}_{\zeta^{n_2}}\ \ \ \  \cancel b \ \cancel b \ \ldots \ \cancel b \ \ \ \ \ b\ b\ \ldots \ \ b \} \\ \nonumber
\zeta_3^{n_3} T_{ab}=&\{a \ a\  \ldots a \  \ \ \ a  \ a \ \ldots \ a \ \ \ \overbrace{\cancel a\ \cancel a\ \ldots \ \cancel a\ \ \ \  \cancel a\ \cancel a\ \ldots \ \cancel a \}}^{\zeta^{n_3}} \nonumber \\
 \times & \{b \  b \ \ldots \  b\ \ \ \ \cancel b\ \cancel b\ \ldots \ \cancel b \ \ \ \ 
\underbrace{ \cancel b \ \cancel b \ \ldots \ \cancel b \ \ \ \ \ b\ b\ \ldots \ b}_{\zeta^{n_3}}\}
\end{align}
Hence $T_{ab}(m_1, m_2, m_3; n_1, n_2, n_3)= \zeta^{n_1}_1 \zeta^{n_2}_2 \zeta^{n_3}_3T_{ab}(m_1, m_2, m_3)$ is defined by 6 integer parameters (the same as the number of degrees of freedom in the state space $\mathbb S^6$ of two qubits in quantum theory, modulo global phase). The frequency of occurrence of the combinations $(a, b)$, $(a, \cancel b)$, $(\cancel a, b)$ and $(\cancel a, \cancel b)$ are determined by $m_1$, $m_2$ and $m_3$ and are independent of $n_1$, $n_2$ and $n_3$. 
Generalising the correspondence (\ref{1qubit}) let us write
\be
\label{2qubit}
\mathscr T_{ab}(m_1, m_2, m_3; n_1, n_2, n_3)\equiv
|\psi_{ab}\rangle=
\cos \frac{\theta_1}{2}|a\rangle
|\psi_b(\theta_2, \phi_2)\rangle
+ e^{i \phi_1} \sin \frac{\theta_1}{2} |\cancel{a}\rangle |\psi_b(\theta_3, \phi_3)\rangle
\ee
where
\begin{align}
\frac{N-m_1}{N}&=\cos^2 \frac{\theta_1}{2};\ \ \ 
\frac{N-m_2}{N-m_1}=\cos^2 \frac{\theta_2}{2};\ \ \ 
\frac{m_1-m_3}{m_1}=\cos^2 \frac{\theta_3}{2};\ \ \ 
\nonumber \\
\frac{n_1}{N}&=\frac{\phi_1}{2\pi};\ \ \ 
\frac{n_2}{N-m_1}=\frac{\phi_2}{2\pi};\ \ 
\frac{n_3}{m_1}=\frac{\phi_3}{2\pi}
\end{align}
As before, making use of the correspondence between quaternions and rotations in space, the angles in (\ref{2qubit}) can be related to orientations in physical space. In addition, the squared amplitudes of the Hilbert tensor product (\ref{2qubit}) corresponds to the frequency of occurrences of $\{a, b\}$, $\{\cancel a, b\}$, $\{a, \cancel b \}$ and $\{\cancel a, \cancel b  \}$ respectively. As before, $|\psi_{ab} \rangle$ represents an equivalence class $\mathscr T_{ab}$ of bit strings such that if $T_a \times T_b$ and $T'_a \times T'_b$ are two members of $\mathscr T_{ab}$, then there exists a permutation $P$ such that $T'_a=P\circ T_a$ and $T'_b=P\circ T_b$ for all $m_j$, $n_j$ - as above, corresponding to a global phase transformation. No matter how large is $N$, it is necessarily the case that 
\be
\cos \theta_j \in \mathbb Q; \ \ \ \frac{\phi_j}{2\pi} \in \mathbb Q\ \ \ \ 1\le j \le 3
\ee

As a particular example, consider
\begin{align}
\label{entangledstate}
T_{ab}(\frac{N}{2}, N-m, m) &=\{\overbrace{ a \ a\  \ldots \  a\ \ \  a \  a \ \ldots \  a}^{\frac{N}{2}}\ \ \ \overbrace{\cancel a\ \cancel a\ \ldots \ \cancel a\ \ \cancel a\ \cancel a\ \ldots \ \cancel a}^{\frac{N}{2}} \}\nonumber \\
& \times \{\underbrace{ b\ b\ \ldots \ b}_{m}\ \ \ \ \underbrace{\cancel b \  \cancel b \ \ldots \  \cancel b}_{\frac{N}{2}-m}\ \ \ \ 
\underbrace{\cancel b \ \cancel b \ \ldots \ \cancel b}_{m}\ \ \ \underbrace{b\ b\ \ldots \ b}_{\frac{N}{2}-m}
\} \nonumber \\
&= T_a(\frac{N}{2},0) \times \zeta^{\frac{N}{2}-m}\ T_b(\frac{N}{2},0)
\end{align}
It can be seen that the parameter $m$ determines the correlation between the two bit strings: for $m=0$ they are perfectly anti-correlated, for $m=N$ they are perfectly correlated. Then
\be
\label{bellstate}
T_{ab}(\frac{N}{2}, N-m, m) \equiv \frac{|a\rangle |\cancel b\rangle - |\cancel a\rangle |b\rangle}{\sqrt 2},
\ee
describes the Bell State where the two detectors are orientated at a relative orientation of $\theta$, where $\cos \theta = 1-2m/N$. In Section \ref{implications}, the elements of the two bit strings in (\ref{entangledstate}) are interpreted as deterministic spin measurement outcomes $A_X(\lambda)$, $B_Y(\lambda)$ in a CHSH experiment, where the hidden-variable $\lambda$ represents the position of an element $A_X(\lambda)$ or $B_Y(\lambda)$ on the bit string. 

The correspondence (\ref{2qubit}) can be generalised so that the Cartesian product $T_{ab\ldots d}=T_a \times T_b \times \ldots \times T_d$ comprising $J$ bit strings, is defined by $2^{J+1}-2$ integer parameters, the same number of degrees of freedom associated with an $J$-qubit state in quantum theory. For example with $J=3$, 
\begin{align}
\label{c2}
&T_{abc}(m_1, m_2, \ldots m_7; \ldots) =\\
&\{\overbrace{ a \ \ a\  \ \ldots \ \  a\ \ \ \ \ a \ \ \  a \ \ \ldots \ \  \ a}^{N-m_1}\ \ \ \ \ \ \ \ \ \overbrace{\cancel a\ \ \ \cancel a\ \ \ \ldots \ \ \ \cancel a\ \ \ \ \cancel a\ \ \ \cancel a\ \ \ \ldots \ \ \ \cancel a}^{m_1} \}
\nonumber \\
 \times &\{\underbrace{b  \ \ b \ \ \ldots \  \ b}_{N-m_2}\ \ \ \ \ \ \underbrace{\cancel b\  \ \ \cancel b\ \ \ \ldots \  \ \ \cancel b}_{m_2-m_1}\ \ \ \ \ \ \ \ 
\underbrace{\cancel b \ \ \ \cancel b \ \ \ \ldots \ \ \ \cancel b}_{m_3}\ \ \ \ \ \underbrace{b\ \ \ b\ \ \ \ldots \ \ \ b}_{m_1-m_3} \} 
\nonumber \\
 \times &\{\underbrace{c \ldots c}_{N-m_4} \underbrace{\cancel c \ldots \cancel c}_{m_4-m_2}\ \ \ \ 
\underbrace{c \ldots c}_{m_2-m_5} \underbrace{\cancel c \ldots \cancel c}_{m_5-m_1}\ \ \ \ \ \ \ \ \underbrace{c \ldots c}_{m_6} \ \underbrace{\cancel c \ldots \cancel c}_{m_3-m_6}\ \ \ \ \ 
\underbrace{c \ldots c}_{m_7} \underbrace{\cancel c \ldots \cancel c}_{m_1-m_3-m_7}\} \nonumber \\ 
\end{align}
with 7 additional `complex' phases defined by permutation operators
\begin{align}
\label{c3}
&T_{abc}(\ldots ; n_1, n_2, \ldots n_7) =\\
&\{\overbrace{ a \ \ a\  \ \ldots \ \  a\ \ \ \ \ a \ \ \  a \ \ \ldots \ \  \ a \ \ \ \ \ \ \ \ \ \cancel a\ \ \ \cancel a\ \ \ \ldots \ \ \ \cancel a\ \ \ \ \cancel a\ \ \ \cancel a\ \ \ \ldots \ \ \ \cancel a}^{\zeta^{n_1}} \}
\nonumber \\
 \times &\{\underbrace{\cancel b  \ \ \cancel b \ \ \ldots \  \ \cancel b \ \ \ \ \ \ b\  \ \ b\ \ \ \ldots \  \ \ b}_{\zeta^{n_2}}\ \ \ \ \ \ \ \ 
\underbrace{b \ \ \ b \ \ \ \ldots \ \ \ b \ \ \ \ \ \cancel b\ \ \ \cancel b\ \ \ \ldots \ \ \ \cancel b}_{\zeta^{n_3}} \} 
\nonumber \\
 \times &\{\underbrace{c \ldots c\  \cancel c \ldots \cancel c}_{\zeta^{n_4}}\ \ \ \ \ \ \
\underbrace{c \ldots c \ \cancel c \ldots \cancel c}_{\zeta^{n_5}}\ \ \ \ \ \ \ \ \ \underbrace{c \ldots c \ \cancel c \ldots \  \cancel c}_{\zeta^{n_6}}\ \ \ \ \ \ 
\underbrace{c \ldots c\  \cancel c \ldots \cancel c}_{\zeta^{n_7}}\} \nonumber \\ 
\end{align}
Here the elements of all three bit strings are permuted within the bounds of the under- or over-brace. This leads to the $14=2^4-2$ degrees of freedom for $T_{abc}(m_1\ldots m_7; n_1\ldots n_7)$. The generalisation to arbitrary $J$ is straightforward. 

As before, write $\mathscr T_{ab \ldots d} \equiv |\psi_{ab \ldots d} \rangle$ where $\equiv$ denotes equality modulo a global bijection of $T_{ab\ldots d}$. This constitutes a major difference with quantum theory where the state space of $J$ qubits is $\mathbb S^{2^{J+1}-2}$. The exponentially larger dimension of the state space in quantum theory arises from the continuum continuity constraints. Are there any experimental consequences of such a difference? Whereas in quantum theory $J$ can be arbitrarily large, in this finite theory, the finite length $N=2^M$ of a bit string imposes an upper limit $J=M$ on the number of bit strings and hence qubits that can be correlated in this way. It can be noted that as the number of qubits tends to the maximum number $M$, then the correlations between entangled qubits becomes increasingly sensitive to the small sample sizes. In this way, it is not the case that entanglement simply ceases at $M$ qubits, it becomes an increasingly noisy and unreliable resource as the number of qubits tends to $M$. Consistent with this, a collection of $J$ qubits where $J>M$ cannot behave purely quantum mechanically and becomes more and more classical the larger is $J$. In this sense, Bohr's delineation between the classical and quantum worlds (as embodied in the Copenhagen interpretation) appears as an emergent concept within this finite framework. 

\subsection{The Bell Theorem}
\label{Bell}

As summarised below, the invariant set model provides new perspectives on the issues of free choice and locality for deterministic interpretations of the Bell Theorem. To see this we first establish some basic notation.  A pair of entangled particles, represented by some specific value of some supplementary variable $\lambda$, have been produced by a common source, are spatially separated, and their spins are each measured by one of two distant experimenters, Alice and Bob. These experimenters can each choose one of two measurement settings, $X \in \{0,1\}$ and $Y \in \{0,1\}$ respectively. The values of $X$ and $Y$ correspond to orientations of finite-precision measurement devices, over which, as with the discussion of the Sten-Gerlach experiment, Alice and Bob have limited control. Despite this limited control, any specific pair of measurements on some specific $\lambda$ corresponds to a pair of precise points on the sphere (see Fig \ref{F:CHSH}). Taken over multiple particle pairs, and for reasonably precise measuring devices, these pairs of points will lie within correspondingly small neighbourhoods $\mathcal X_0$, $\mathcal X_1$, $\mathcal Y_0$, $\mathcal Y_1$ on the sphere. Once performed, the measurements yield outcomes $A \in \{0,1\}$ and $B\in \{0,1\}$ respectively. We assume, as is conventional, some deterministic theory such that $A=A_{XY}(\lambda)$, $B=B_{YX}(\lambda)$ are deterministic formulae. Hence, with $\Lambda$ a finite sample space of supplementary variables,
\begin{equation}
\label{prob}
E(AB|XY)= \sum_{\lambda \in \Lambda} A_{XY}(\lambda) B_{YX}(\lambda)\;  p(\lambda | XY) 
\end{equation}
denotes an expectation value for the product $AB$, and where $p(\lambda|XY)$ denotes a probability function on $\lambda \in \Lambda$. For our finite theory, $p(\lambda|XY)$ can be simply taken as some normalisation constant $m_{XY}$. It is crucially important in all that follows to note that the expectation value on the left hand side of (\ref{prob}) is defined from specific points in the neighbourhoods $\mathcal X_0$ etc, but is not defined over the whole of these neighbourhoods. We now define:
\begin{itemize}
\item \textbf{Statistical Independence}: $p(\lambda | XY)=p(\lambda)$, i.e. $m_{XY}=m_0$. 
\item \textbf{Factorisation}: $A_{XY}(\lambda)B_{YX}(\lambda) = A_X(\lambda)B_Y(\lambda)$
\end{itemize}
Bell's Theorem \cite{Bell} implies that if a putative deterministic theory satisfies Statistical Independence and Factorisation, then it must satisfy the Bell Inequality and be inconsistent with experiment. In analyses of the Bell Theorem, Statistical Independence and Factorisation are conventionally taken as definitions of free choice and local causality, respectively. A theory which violates local causality is referred to as Superdeterministic \cite{HossenfelderPalmer}. Below we show that in the invariant set model, both Statistical Independence and Factorisation are violated. It will, however, be argued that such violation need not imply a violation of either free choice or local causality if one returns to the physical space-time concepts underpinning free choice and local causality as discussed by Bell and others. This motivates modified definitions of Statistical Independence and Factorisation, which are not violated in the invariant set model. 
%..........................................................................................................................................
\begin{figure}
\centering
\includegraphics[scale=0.3]{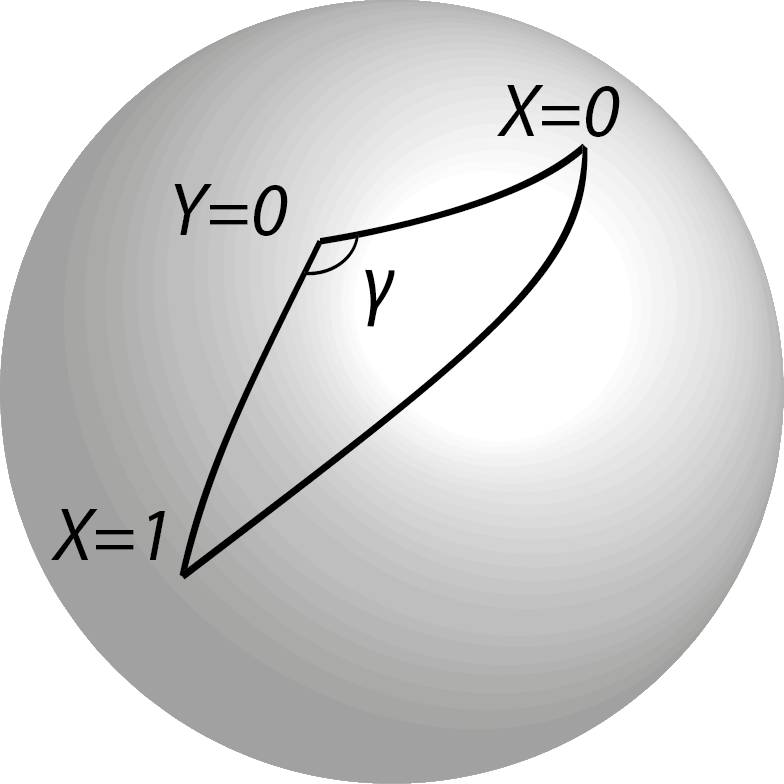}
\caption{\emph{In general it is impossible for all the cosines of the angular lengths of all three sides of the spherical triangle $\triangle(X=0, X=1, Y=0)$ to be rational, and the internal angle $\gamma$ to be a rational multiple of $\pi$. That is, the finiteness conditions for Hilbert states cannot be satisfied for a counterfactual measurement $X=1$, $Y=0$ on a particular particle pair, when it is satisfied for a realisable measurement $X=0$, $Y=0$ on the same particle pair. Because of this, Statistical Independence and Factorisation are violated in the invariant set model. However, precisely because the counterfactual state violates the finiteness conditions and therefore does not correspond to a state of physical reality on $I_U$, the invariant set model does not satisfy `Statistical Independence on $I_U$' and `Factorisation on $I_U$'. Hence whether invariant set theory is local or nonlocal depends critically on whether locality should be expressible entirely in terms of changes to quantities defined in space-time, or should have unrestricted access to counterfactual states in potentially non-onic parts of state space.}}
\label{F:CHSH}
\end{figure}
%..........................................................................................................................................

For some particular $\lambda$, suppose in reality Alice chooses $X=0$ and Bob chooses $Y=0$ so that $\rho (\lambda |00)=m_{00} \ne 0$. In invariant set theory, this implies that the triple $(\lambda, X=0, Y=0)$ corresponds to a point on $I_U$. According to the analysis in Section \ref{entangle}, this is possible if $\cos \theta_{00}$ is rational (where $\theta_{XY}$ denotes the angular distance between the precise points associated with $X$ and $Y$ on the sphere). Now, when we ask, for example, whether $\rho (\lambda |10)=m_{00}$, then noting that this equation refers to the same specific $\lambda$, the conditional $Y=0$ is not referring to some separate point in the nominal neighbourhood $\mathcal Y_0$, but to the specific point $Y=0$ on the Bloch sphere associated with the particular measurement Bob made for that specific $\lambda$. According to the analysis in Section \ref{entangle}, this implies that $\cos \theta_{10}$ must also be rational. 

However, this implies an inconsistency. Although the choice $X=1$ for that specific $\lambda$ is counterfactual (since we have assumed Alice actually chose $X=0$), whatever the specific measurement orientation associated with $X=1$, it should always be possible for Alice, having measured the particle with respect to $X=0$, to feed that same particle back into the measuring apparatus, now set to $X=1$. Hence, following the analysis of single qubit physics, the cosine of the angle between $X=0$ and $X=1$ must also be rational. However, now we have a contradiction, identical in form to the analysis of non-commuting observables in the sequential Stern-Gerlach analysis. Specifically, if two sides of the triangle $\triangle(X=0, X=1, Y=0)$ are rational, the third side cannot be. (One can ask whether all of $X=0$, $X=1$ and $Y=0$ could lie on a single great circle precisely. However, as discussed for the sequential Stern-Gerlach experiment, if there is a maximum finite resolution to the measuring system, it will be impossible for Alice and Bob to control their measurement orientations sufficiently for $X=0$, $X=1$ and $Y=0$ to lie \emph{precisely} on a great circle.)

Hence one must conclude that $\rho(\lambda |10)=0\ne m_{00}$. The counterfactual perturbation that takes $X=0$ to $X=1$, keeping $\lambda$ and Bob's measurement orientation fixed, takes a state $U$ of the universe off $I_U$ and hence to a state of physical unreality. Alice and Bob do not have sufficient control on their measuring instruments to set them so that the cosine of their relative orientation is irrational - which is to say that the triple $(\lambda, X=1, Y=0)$ is not expressible as an experiment in space-time. The counterfactual experiments considered here takes states off the invariant set to a zero of the invariant (Haar) measure of $I_U$. 

A similar argument shows that Factorisation is violated. By the argument above, if $B_{00}(\lambda)$ has some definite value, then $B_{01}(\lambda)$ does not. Hence Factorisation fails. 

The key questions that arise from these conclusions are whether the invariant set model violates free choice and local causality. Clearly, if one defines free choice and local causality by Statistical Independence and Factorisation, then invariant set theory must be nonlocal. However, this raises a subtle but important point which does not arise when considering more simplistic types of hidden-variable model. Let us first focus on the issue of locality. Bell himself wrote \cite{Bell} `..my intuitive notion of local causality is that events in [a space-time region] 2 should not be 'causes' of events in [a spacelike separated space-time region] 1, and vice versa.' More recently, Wharton and Argaman approach locality in terms of `screening regions' in space time \cite{Wharton}. Their  `Continuity of Action' (on which the notion of locality is derived) is motivated by considering whether `changes in [space-time region] 2 can be associated with changes in [space-time region] 1 without being associated with changes within [some intermediate shielding region] S'. Again, the guiding concept behind these notions of local causality is that of causal relationships between events in space time. 

However, as discussed, the perturbations which demonstrate that Statistical Independence and Factorisation are violated in the invariant set model, take a point on $I_U$ and perturb it to a point which does not lie on $I_U$. By the postulates of invariant set theory (see Section \ref{invariant}), such a perturbed point does not correspond to a state of physical reality and is therefore not associated with some perturbation or change to events or fields in space-time. That is to say, the violations of Statistical Independence and Factorisation in invariant set theory are not associated with changes which have any expression in space time, they are instead associated with counterfactual changes, defined in state space. Hence, if we insist that the notions of free choice and local causality are strictly associated with changes which are expressible in space time, then invariant set theory is locally causal and does not violate free choice. 

To make this point more precisely we need to modify the definitions of Statistical Independence and Factorisation so that they only refer to changes or perturbations which are expressible in space-time. This is easily done by the following:
\begin{itemize}
\item \textbf{Statistical Independence on $I_U$}: $p(\lambda | XY)=p(\lambda)$ for triples $(\lambda, X,Y)$ corresponding to states on $I_U$.
\item \textbf{Factorisation on $I_U$}: $A_{XY}(\lambda)B_{YX}(\lambda) = A_X(\lambda)B_Y(\lambda)$ for triples $(\lambda, X,Y)$ corresponding to states on $I_U$
\end{itemize}
The `Factorisation on $I_U$' condition is therefore a description of locality, if locality is assumed to refer only to relationships between events and fields expressible in space time. A similar analysis applies to the notion of free will. If free will is defined by the notion that `I could have done otherwise', then free will necessarily involves counterfactual worlds which have no expression in space-time. However, free will can also be defined by processes which do occur in space time: an agent is free if there are no constraints that might otherwise prevent her from doing as she wishes (this is the `compatibilist' definition of free will \cite{Kane}). 'Statistical Independence on $I_U$' expresses this more operational notion of free will/free choice. 

Invariant set theory satisfies `Statistical Independence on $I_U$' and `Factorisation on $I_U$'. To be more explicit, $\lambda \in \Lambda$ associated with a state on $I_U$ - by the analysis above - belongs to one of two classes: either 1) where $X=0, Y=0$, or $X=1, Y=1$; or 2) where $X=0, Y=1$, or $X=1,Y=0$. For some $\lambda$ belonging to one of these two classes (say class 1)), knowing the value of $X$ (say $X=0$) fixes the value of $Y$ ($Y=0$), and hence the argument $Y$ in $A_{XY}(\lambda)$ is redundant, which can therefore be written $A_X(\lambda)$. Similarly for $B_{YX}(\lambda)$ where, given the value of $Y$, the argument $X$ is redundant. Hence, with a $\lambda$ belonging to one of these two classes we can write $A_{XY}(\lambda)B_{YX}(\lambda)=A_X(\lambda)B_Y(\lambda)$ which implies `Factorisation on $I_U$'. In this sense, the supplementary variable $\lambda$ characterises properties of $I_U$ which are themselves not discoverable by experiment or computation in space-time. This of course is consistent with the notion of $\lambda$ being `hidden'. However, rather than postulate hiddenness by fiat, here hiddenness is consistent with (and hence emergent from) the notion that, in a formal sense, almost all non-trivial properties of fractal invariant sets are non-computable \cite{Blum} \cite{Dube:1993}, or in the case of finite representations, computationally irreducible \cite{Wolfram} i.e. such properties cannot be found by reduced-precision computation. It can be noted that Penrose \cite{Penrose:1989} has for many years speculated that non-computability may play a central role in deterministic reformulations of quantum theory; invariant set theory provides a specific illustration of this notion. More recently Walleczek \cite{Walleczek} has concluded that the inaccessibility to the experimenter of the complete quantum state may be a consequence of non-computability at some deeper ontological level. Such non-computability is consistent with the fact that by finite precision, it is impossible for an experimenter to determine whether a putative Hilbert state is associated with finite squared amplitudes and phases, or not. More specifically, $\lambda$ describes how the particle relates to the rest of $U$, embodied in the reference direction $(\theta_{\text{ref}}, \phi_{\text{ref}})$ (Section \ref{SG}), and about the specific helical trajectory $J$, $1\le J\le N$ associated with the current aeon of the universe. Despite, this, some probabilistic facts are known to the local experimenter. For example, as with the Stern-Gerlach experiment, taken over a random set of unknown $\lambda$ values, one can assume that $(\theta_{\text{ref}}, \phi_{\text{ref}})$ is random with respect to a uniform probability distribution on the sphere. With this, the probability of up/down measurement outcomes by either Alice or Bob will be equal to 1/2, consistent with (\ref{bellstate}). 

As discussed in Section \ref{entangle}, invariant set theory is based on Hilbert tensor products. In particular, the singlet Bell state is represented by a pair of $N$-bit strings (c.f. equation (\ref{bellstate})) and hence violates the Bell inequality exactly as does quantum theory. Hence, another way to address this question of locality vs nonlocality is to ask whether quantum theory is itself local or nonlocal. Here the community is divided. Whilst some researchers (e.g. \cite{Brunner}, \cite{Maudlin}) conclude quantum theory is nonlocal, others are adamant it is not (e.g. \cite{Griffiths:2019}, \cite{Khrennikov}). These latter papers emphasise that the violation of the Bell inequality arises in quantum theory from the incompatibility of non-commuting observables, and not from nonlocality \emph{per se}. In invariant set theory, the incompatibility of non-commuting observables in quantum theory has been represented by number-theoretic incommensurateness associated with perturbations in state space which would take states on $I_U$, off $I_U$. Again, this is not inconsistent with local causality if the latter is defined purely in terms of causal relationships between events or fields in space time. 

In conclusion:

\begin{itemize}

\item Invariant set theory violates both Factorisation and Statistical Independence. If these conditions define local causality and free choice, then invariant set theory is nonlocal. 

\item However, if local causality is defined purely as a restriction on the relationships between space-time events (and more generally in terms of fields in space-time), then locality in invariant set theory should be defined by `Factorisation on $I_U$', in terms of which invariant set theory is local. Bell's intuitive formulation of local causality (and other more recent ones) is based on causal relationships between space-time events. 

\item Similarly, if free choice is defined as an absence of restrictions, expressible in space-time, on an experimenter's desires or wishes (e.g. associated with neuronal action in the experiment's brain), then free choice in invariant set theory should be defined by `Statistical Independence on $I_U$', in terms of which invariant set theory does not violate free choice.  Compatibilist formulations of free will  (i.e. those consistent with determinism) are defined in terms of constraints in space time. Non-compatibilist definitions (`I could have done otherwise') are not. 

\item Researchers who believe that quantum theory is local do so because the analysis of the Bell Inequality within a purely quantum theoretic framework indicates that the key reason the Bell inequality is violated is the incompatibility of non-commuting observables (rather than nonlocality \emph{per se}). Such incompatibility arises from the non-classical structure of quantum theory in configuration space and not from violation of causality in space-time. Consistent with this, invariant set theory provides a deterministic account of such incompatibility in terms of a non-classical fractal invariant set structure on state space, associated with number-theoretic restrictions on Hilbert states. Perturbations which take states off the invariant set have no expression as changes in space-time. Hence, insofar as quantum theory is local, so is invariant set theory. Insofar as quantum theory is nonlocal, so is invariant set theory. 

\end{itemize}

\subsection{GHZ}
\label{GHZ}

The arguments above can also be straightforwardly applied to interpret the GHZ state \cite{GHZ} 
\be
|\psi_{\text{GHZ}}\rangle=\frac{1}{\sqrt 2}(|v_A\rangle|v_B\rangle |v_C\rangle +|h_A\rangle|h_B\rangle |h_C\rangle)
\ee
realistically. Here we consider a polarisation-entangled state comprising three photons where $v$ and $h$ denote vertical and horizontal polarisation. As before, it is possible to choose to make linear polarisation measurements on any of the three photons at an angle $\phi$ to the $v/h$ axis, providing $\cos^2 \phi$ and hence $\cos 2 \phi$ is rational. The corresponding unitary transformation is
\be
\label{GHZtrans}
\begin{pmatrix}
v' \\
h'
\end{pmatrix}
=
\begin{pmatrix}
\cos \phi & - \sin \phi \\
\sin \phi & \cos \phi
\end{pmatrix}
\begin{pmatrix}
v \\
h
\end{pmatrix}
\ee
It is also possible to choose to make circular polarisation measurements on any of the photons, whence the corresponding unitary transformation is  
\be
\label{GHZtrans2}
\begin{pmatrix}
L\\
R
\end{pmatrix}
=
\begin{pmatrix}
1 & -i \\
1 & i
\end{pmatrix}
\begin{pmatrix}
v \\
h
\end{pmatrix}\ \ \ \ \ \ \ \ \ \ 
\ee
By considering the case $\phi \approx 45^\circ$ and both linear and circular polarisation possibilities for the photons, it is well known that it is impossible to explain measurement correlations on the GHZ state with a conventional classical local hidden-variable theory. However, it is possible to explain these correlations realistically using the finite theory developed above. To see this, consider, say, the second photon. Suppose in reality the experimenter measures this photon relative to the $v'/h'$ basis. Let us ask the counterfactual question: What would she have measured had she measured this photon relative to the $L/R$ basis? To answer this question, note from (\ref{GHZtrans}) and (\ref{GHZtrans2})
\begin{align}
\begin{pmatrix}
L\\
R
\end{pmatrix}
=&
\begin{pmatrix}
1 & -i \\
1 & i
\end{pmatrix}
\begin{pmatrix}
\cos \phi & \sin \phi \\
-\sin \phi & \cos \phi
\end{pmatrix}
\begin{pmatrix}
v' \\
h'
\end{pmatrix} \nonumber \\
=&
\begin{pmatrix}
e^{i \phi} & e^{i(\phi-\pi/2)} \\
e^{-i\phi} & e^{-i(\phi-\pi/2)}
\end{pmatrix}
\begin{pmatrix}
v' \\
h'
\end{pmatrix}
\end{align}
As above, although $\phi$ may equal $45^{\circ}$ to any nominal accuracy, it cannot equal $45^{\circ}$ \emph{precisely}. Hence if $U$ denotes a universe where the experimenter chose to measure the linear polarisation of one of these photons - implying that $\cos  2\phi$ must be rational - then the Hilbert state corresponding a measurement of circular polarisation on this same photon, is undefined, because if $\cos  2\phi$ is rational, then $\phi$ cannot be a rational multiple of $2 \pi$. Conversely, if the experimenter chose to measure circular polarisation, then she could not have measured linear polarisation. In these cases, the counterfactual experiments cannot lie on $I_U$ and therefore do not correspond to states of reality. Hence the argument fails that would otherwise disallow a realistic interpretation of GHZ.

\section{Discussion}
\label{discussion}

\begin{quote}
`\emph{The infinite is nowhere to be found in reality, no matter what experiences, observations, and knowledge are appealed to.} David Hilbert \cite{Hilbert}.
\bigskip
\end{quote}
Many scientists will be sympathetic to the implications of Hilbert's observation - that if the infinite is nowhere to be found in reality, it should neither be found in our descriptions of reality.  However, this is problematic for quantum theory where, even for finite dimensional systems, the notion of the infinitesimal plays a vital role (\cite{Hardy:2004}). 

In this paper we have developed a finite deterministic alternative to quantum theory. Quantum theory does not arise as the smooth limit of invariant set theory as $N \rightarrow \infty$, but is a singular limit at $N=\infty$. In this paper we have focussed entirely on the discretised kinematics of finite dimensional quantum systems, showing how a constructive description of complex Hilbert states on the Bloch Sphere leads to novel realistic interpretations of what are often seen as the defining characteristics of quantum physics:  quantum complementarity uncertainty relationships and entanglement. A new perspective on the issue of nonlocality in quantum physics has been presented. These novel interpretations arise from number-theoretic properties of trigonometric functions. Such properties are not part of the standard quantum theoretic canon, possibly suggesting the existence of new physics. 

Readers may well wonder the extent to which the current theory is truly finite, being based on fractal geometry. A fractal can be defined in term of some intersection 
\be
I=\bigcup_{j=1}^{\infty} I_j
\ee
of fractal iterates $I_j$ of a fractal invariant set $I$. However, all the results in this paper apply arbitrarily well to finite approximations
\be
I_J=\bigcup_{j=1}^{J} I_j
\ee
of $I$ for large enough $J$. In physical terms, such a finite-$J$ approximation would refer to a universe which evolves through a large but finite number of aeons, before repeating. 

We conclude with a central question for any putative theory of physics: What are its experimental implications? Based on the analysis in this paper, the most important of these is the fact that only a finite number $\log_2 N$ of qubits can be maximally entangled. Now it has been speculated that gravity may be an intrinsically decoherent phenomenon \cite{Penrose:1994} \cite{Diosi:1989}, in which case self-gravitation may also put an upper limit the number of entangled qubits. This raises the question of whether the presumed fractal-like geometry which underpins the discretised nature of Hilbert Space may itself provide the basis for a description of gravitational effects. We have, for example, described the measurement process in terms of a clustering of state space trajectories into distinct regions of state space. It is conceivable that this clustering process is itself a description of universal gravitational attraction, and that the assumed largeness of the parameter $N$ may reflect the relative weakness of gravity.

If this is so, then this discretised theory would predict that gravity cannot be an entanglement witness \cite{MarlettoVedral}. In this sense, the proposed approach shares properties in common with collapse theories of quantum physics, in which gravitation plays an important role in effecting collapse of the wavefunction. On the other hand, the proposed discretised theory is \emph{not} a collapse theory, since all superposed Hilbert states on $I_U$ can be interpreted constructively in terms of ensembles of trajectories of $I_U$. 

In the (singular) limit at $N=\infty$, the set of discretised complex Hilbert states becomes the full complex Hilbert Space, whose synthesis with general relativity theory is known to be deeply problematic. By taking a step back towards finite $N$ it may be much easier to synthesise a finite theory of quantum physics with the deterministic, nonlinear causal theory of general relativity. 

\section*{Acknowledgements}

I acknowledge helpful input from from Drs Roger Colbeck, Lucien Hardy, Sabine Hossenfelder, Andrei Khrennikov, Oliver Reardon-Smith, Tony Sudbery, Felix Tennie, Vlatko Vedral and from three anonymous reviewers. This research was supported by a Royal Society Research Professorship. 
\bibliography{mybibliography}
\end{document}